\documentclass[11pt]{article}
\usepackage{amssymb}
\usepackage{amsmath}
\usepackage{graphicx,color}
\usepackage{amsfonts}
\usepackage[ignoreall]{geometry}
\usepackage{natbib}
\usepackage{enumitem}
\usepackage{setspace}

\usepackage{relsize}

\newcommand{\ttt}{\boldsymbol \theta}

\newcommand{\OO}{\mbox{$\mathbf O$}}
\newcommand{\bb}{\mbox{$\mathbf b$}}
\newcommand{\CC}{\mbox{$\mathbf C$}}

\newcommand{\FF}{\mbox{$\mathbf F$}}

\newcommand{\XX}{\mbox{$\mathbf X$}}
\newcommand{\xx}{\mbox{$\mathbf x$}}
\newcommand{\V}{\mbox{$\mathrm {Var}$}}

\newcommand{\qq}{\mbox{$\mathbf q$}}

\newcommand{\argmin}{\operatornamewithlimits{arg\,min}}

\newcommand{\MFL}{\mbox{$\mathfrak L$}}

\newcommand{\II}{\mbox{$\mathcal{I}^*$}}
\newcommand{\PP}{\mbox{$\mathbf{P}$}}
\newcommand{\MM}{{\mathcal{M}}}
\newcommand{\MD}{{\mathcal{D}}}
\newcommand{\ML}{{\mathcal{L}}}
\newcommand{\MS}{{\mathcal{S}^*}}

\newcommand{\UM}{{\hat{U}_{1,\MM_2}}}
\newcommand{\UA}{{\tilde{U}_{1,\MM_2}}}
\newcommand{\UCM}{{\hat{U}_{2,\MM_2}}}
\newcommand{\DM}{{\hat{D}_{1,\MM_2}}}
\newcommand{\SM}{{\hat{S}_{\MM_2}}}
\newcommand{\LM}{{\hat{L}_{\MM_2}}}
\newcommand{\SA}{{\tilde{S}_{\MM_2}}}
\newcommand{\LA}{{\tilde{L}_{\MM_2}}}
\newcommand{\DDD}{{\hat{D}_{1,\hat{\MD}}}}
\newcommand{\LD}{{\hat{L}_{\hat{\MD}}}}
\newcommand{\SD}{{\hat{S}_{\hat{\MD}}}}

\newcommand{\DSB}{{\tilde{S}_{{\MS}^{\bot}  }}}
\newcommand{\DLB}{ \tilde{L}_{(T_{\tilde{L}}\MFL)^{\bot}}  }
\newcommand{\UAtwo}{{{\tilde{U}_{2,\mathcal{M}_2}  }}}
\newcommand{\UMtwo}{{{\hat{U}_{2,\mathcal{M}_2}  }}}
\newcommand{\DAtwo}{{{\tilde{D}_{2,\mathcal{M}_2}  }}}
\newcommand{\DAone}{{{\tilde{D}_{1,\mathcal{M}_2}  }}}
\newcommand{\ce}{{e^{-N}}}

\newcommand{\E}{\mathbb{E}}

\newcommand{\cs}{{\gamma_N^{1-2\eta}}}
\newcommand{\cu}{{\gamma_N^{1-\eta}}}
\newcommand{\cusq}{{\gamma_N^{2(1-\eta)}}}
\newcommand{\cuqb}{{\gamma_N^{3(1-\eta)}}}
\newcommand{\cd}{{\gamma_N^{1-2\eta} }}

\newcommand{\bA}{{\mathbf{A}}}

\newtheorem{theorem}{Theorem}

\newtheorem{lemma}{Lemma}

\newtheorem{proposition}{Proposition}
\newtheorem{remark}{Remark}

\newenvironment{proof}[1][Proof]{\noindent\textbf{#1.} }{\ \rule{0.5em}{0.5em}}

\title{A Fused Latent and Graphical Model for Multivariate Binary Data}
\author{Yunxiao Chen, Xiaoou Li, Jingchen Liu, and Zhiliang Ying\\ Columbia University}

\date{}

\begin{document}
\maketitle
\baselineskip = 20pt

\begin{abstract}
	We consider modeling, inference, and computation  for analyzing multivariate binary data.
We propose a new model that consists of a low dimensional latent variable component and a sparse graphical component.
Our study is motivated by analysis of item response data in cognitive assessment and has applications to many disciplines where item response data are collected. Standard approaches to item response data in cognitive assessment adopt the multidimensional item response theory (IRT) models. However,
human cognition is typically a complicated process and thus may not be adequately described by just a few factors.
Consequently, a low-dimensional latent factor model, such as the multidimensional IRT models, is often insufficient to capture the structure of the data. The proposed model adds a sparse graphical component that captures the remaining ad hoc dependence. It reduces to a multidimensional IRT model when the graphical component becomes degenerate.
Model selection and parameter estimation are carried out simultaneously through construction of a pseudo-likelihood function and properly chosen penalty terms. The convexity of the pseudo-likelihood function allows us to develop an efficient algorithm, while the penalty terms generate a low-dimensional latent component and a sparse graphical structure. Desirable theoretical properties are established under suitable regularity conditions. The method is applied to the revised Eysenck's personality questionnaire, revealing its usefulness in item analysis. Simulation results are reported that show the new method works well in practical situations.

%
\end{abstract}


\section{Introduction} \label{introduction}

	Latent variable models are prevalent in many studies. We consider the context of cognitive assessment that has applications in many disciplines including education, psychology/psychiatry, political sciences, marketing, etc.
	For instance, in educational measurement, students' solutions to test problems are observed to measure their skill levels;
	in psychiatric assessment, patients' responses to diagnostic questions are observed to assess the presence or absence of mental health disorders;
	in political sciences, politicians'  voting behavior reflects their political views; in marketing analysis, consumers' purchase history reflects their preferences.
	A common feature in these studies is that the observed human behaviors are driven by their latent attributes that are often unobservable.
	Latent variable models can be employed in these contexts to describe the relationship between the observed behavior, which is often in the form of responses to items, and the underlying attributes.

	Various linear and nonlinear latent variable models have been studied extensively in the literature \citep[e.g.][]{joreskog1969general,mcdonald1985factor,harman1976modern,rasch1960probabilistic,lord1968statistical,joreskog1973general}.
	In this paper, we focus on one of the widely used nonlinear models for categorical responses, that is, the item response theory (IRT) model. 
	Popular IRT models include the Rasch model \citep{rasch1960probabilistic}, the two-parameter logistic model, and the three-parameter logistic model \citep{birnbaum1968some} that are single-factor models. A natural extension is the multidimensional two-parameter logistic (M2PL) model \citep{mckinley1982use,ReckaseMIRT} assuming a multidimensional latent vector.
	Originated in psychological measurement \citep{rasch1960probabilistic,lord1968statistical},
IRT models have been widely used in other fields for modeling multivariate binary data, such as political voting \citep{bafumi2005practical}, marketing \citep{de2008using}, and health sciences \citep{hays2000item,streiner2014health}.

	In this paper, we use the multidimensional two-parameter logistic model as the starting point. In particular, each observation is a random vector $\XX = (X_1,...,X_J)$ with binary components, $X_j \in \{0,1\}$.
	Associated with each observation is an unobserved continuous latent vector $\ttt \in \mathbb R^K$. The conditional distribution of each response given the latent vector follows a logistic model
\begin{equation}\label{IRT}
f_j(\theta) = \mathbb P( X_j =1| \ttt) = \frac{e^{a_j^\top \ttt + b_j}}{1+ e^{a_j^\top \ttt + b_j}},
\end{equation}
which is known as the \emph{item response function}.
Furthermore, the responses are assumed to be conditionally independent given $\ttt$, that is,
\begin{equation}\label{local}
\mathbb  P(X_1=x_1,...,X_J = x_J | \ttt) = \prod_{j=1}^J \mathbb  P( X_j =x_j| \ttt).
\end{equation}
A prior distribution $\pi$ on $\ttt$ is also imposed.
	
	In recent years, computer-based instruments are becoming prevalent in
educational and psychiatric studies,
where a large number of responses with complex dependence structure are  observed.
	A low-dimensional latent vector is often insufficient to capture all the dependence structure of the responses.
Many contextual factors, such as the item wording and question order, may exert additional influence on the item response \citep{knowles2000does,schwarz1999self,yen1993scaling}.
Moreover, problem solving and task accomplishing are likely to be complicated cognitive processes.
It is conceivable that they cannot be adequately described  by only a few latent attributes. Thus, model lack of fit is often observed in practical analysis \citep[e.g.][]{reise2011challenges,ferrara1999contextual,yen1984effects,yen1993scaling}.
From the technical aspect, a low-dimensional latent variable model is simply not rich enough to capture all the dependence structure of the responses \citep[e.g.][]{sireci1991reliability,chen1997local}.
	Ideally, we wish to include all the factors that influence the cognitive process, that would result in a high-dimensional latent vector.
A factor model with too
many latent variables
can be difficult to estimate and may lack interpretability. Thus, in practice, the dimension of the latent vector $K$ is often kept low in spite of the lack of fit. 	
	
We propose a new model that maintains a low-dimensional latent structure and captures the remaining dependence.	
	We achieve this by including an additional graphical component to describe the dependence that is not explained by the low-dimensional latent vector.
We call it \textit{Fused Latent and Graphical} (FLaG) model. The new model captures two interpretable sources of dependence, i.e. the common dependence through the latent vector and the  ad hoc dependence through a sparse graphical structure.

Figure~\ref{fig:graph1} provides a graphical illustration of the multidimensional IRT model and the FLaG model. The left panel shows a graphical representation of the marginal distribution of the responses, where there is an edge between each pair of responses.
	Under the conditional independence assumption \eqref{local}, there exists a latent vector $\ttt$. If we include $\ttt$ in the graph, then there is no edge among $X$'s.
	Our concern is that the middle graph may be oversimplified and there may not exist a low-dimensional $\ttt$ to achieve such a simple structure.
	The FLaG model (the right panel) is a natural extension.
	There remain edges among $X$'s even if $\ttt$ is included, suggesting that $\ttt$ does not fully explain the dependence among $X$. However, the remaining dependence is substantially reduced compared with the left penal.

\begin{figure}[ht]
\centering
\includegraphics[scale = 0.3]{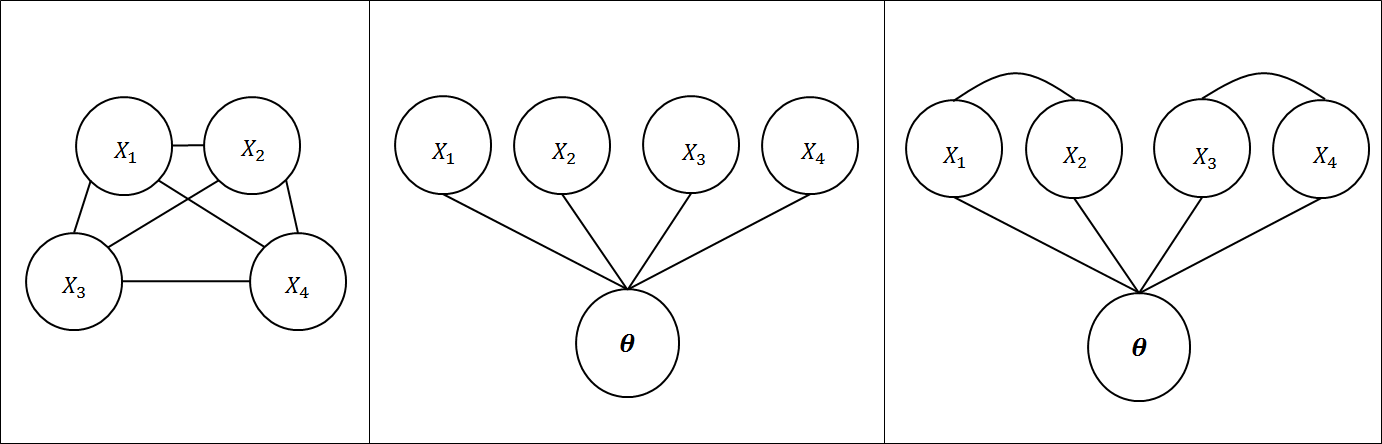}
\caption{A graphical illustration of the multidimensional IRT model and the proposed FLaG model.}
\label{fig:graph1}
\end{figure}

	From the inference viewpoint, it is important to separate the common dependence due to the latent vector and the ad hoc dependence due to the graphical structure. To do so, we make inference based on the following assumptions.
	The variation of responses is mostly characterized by the latent vector.  A low-dimensional latent vector model is largely correct and majority of the dependence among the responses is induced by the common latent vector.
	There is just a small remainder due to the graphical structure. In particular, the conditional graph does not have too many edges. Technical statements of these assumptions will be described in the sequel.
	During the estimation, we assume that neither the dimension of the latent vector $K$ nor the graphical structure is known. We estimate the latent structure and the conditional graph simultaneously by including penalty terms to regularize the dimension of the latent vector and the number of edges of the conditional graph.
	Thus, the resulting model contains a low-dimensional latent vector and a sparse conditional graph.
	
	To model  the graphical component, we adopt an undirected graph  that characterizes the conditional independence/dependence structure among the responses \citep{pearl1988probabilistic, lauritzen1996graphical}.
	In particular, we consider the Ising model originated in physics \citep{ising1925beitrag}. It has been employed to model multivariate binary data in political voting \citep{banerjee2008model} and  climate data \citep{barber2015high}.
	Estimation of the Ising model via regularization has been studied by \cite{hofling2009estimation}, \cite{ravikumar2010high}, and \cite{guo2010joint}.


The proposed  modeling framework is  related to the analysis of decomposing a matrix into low-rank and sparse components \citep{candes2011robust,zhou2010stable,chandrasekaran2011rank} and the statistical inference of a
multivariate Gaussian model whose precision matrix admits the form of a low-rank matrix plus a sparse matrix \citep{chandrasekaran2012latent}.
However, the inference and optimization of the current model are different from the linear case.
We construct a pseudo-likelihood function, based on which a regularized estimator is proposed for simultaneous model selection and parameter estimation. The optimization for the regularized estimator is convex, for which
we develop an efficient algorithm through the alternating direction method of multiplier \citep[ADMM;][]{boyd2011distributed, glowinski1975approximation, gabay1976dual}.


	The rest of this paper is organized as follows. In Section \ref{section: model}, we first provide a brief review of the multidimensional item response theory model and the Ising model.
It is then followed by the introduction of the FLaG model.
Section \ref{sec:est} introduces a pseudo-likelihood function and
 presents the regularized pseudo-likelihood estimator.
 An efficient algorithm is developed and related computational issues are also discussed in Section \ref{sec:comp}. Section \ref{section: simulation} includes simulation studies and a real data analysis.

\section{Fused latent and graphical model} \label{section: model}

\subsection{Two basic models}


	To begin with, we present two commonly used models as the basic building blocks: the multidimensional two-parameter logistic model and the Ising model. We consider that $N$ independent and identically distributed random vectors are observed. We use $\XX_i = (X_{i1},...,X_{iJ})$ to denote the $i$th random observation and $\xx_i = (x_{i1},...,x_{iJ})$ its realization. Furthermore, we use $\XX = (X_1,...,X_J)$ as a generic random vector equal in distribution to each $\XX_i$.
	Throughout this paper, we consider binary observations, that is, each $X_{ij}$ takes values in $\{0,1\}$. For more general types of categorical variables, the analysis can be extended if it can be fit into an exponential family.
	
	Latent variable models assume that there exists an unobserved random vector $\ttt =(\theta_1,...,\theta_K)$ associated with $\XX$, such that the conditional distribution of $\XX$ given $\ttt$ takes a simpler form that is easy to parameterize and estimate.
	For instance, the conditional variance $\V(\XX|\ttt)$ is substantially reduced compared to $\V(\XX)$, in which case the random vector $\XX$ is very close to (or essentially lives on) a low-dimensional manifold generated by $\ttt$.
	Another popular approach is to assume that $\XX$ is conditionally independent given $\ttt$, that is,
	$$f(\xx|\ttt) = \prod_{j=1}^J f_j(x_j|\ttt).$$
	This is also known as the \emph{local independence} assumption that is widely used in cognitive assessment \citep{Embretson}.
	In this case, the dependence among $X_j$'s is fully accounted for by the common latent vector $\ttt$ and the variation of $X_j$ given $\ttt$ is essentially considered as independent random noise.
	
	Latent variable models largely fall into two classes based on the type of $\ttt$: discrete and continuous. In this paper, we consider the latter that $\ttt\in \mathbb R^K$ is a $K$-dimensional continuous random vector.
	The multidimensional item response theory model is a popular class of nonlinear latent variable models. The conditional distribution of each $X_j$ given $\ttt$ admits the form of a generalized linear model. In the case of binary data, the most popular is the multivariate 2-parameter logistic model (M2PL)
\begin{equation}\label{IRF}
 \mathbb P( X_j =1| \ttt) = \frac{e^{a_j^\top \ttt + b_j}}{1+ e^{a_j^\top \ttt + b_j}},
\end{equation}
where $a_j = (a_{j1},...,a_{jK})^\top \in \mathbb R^K $ is the loading vector of the latent vector and $b_j\in \mathbb R$ controls the marginal probability of $X_j$.
	The above probability as a function of $\ttt$ is also known as the \emph{item response function}.
	Furthermore, the responses are assumed to be independent conditional on the latent vector $\ttt$, that is,
\begin{equation*}
\tilde f_{\ttt} (\xx) =\mathbb  P(X_1=x_1,...,X_J = x_J | \ttt) = \prod_{j=1}^J \mathbb  P( X_j =x_j| \ttt).
\end{equation*}
	In addition, a prior distribution $\pi$ is imposed and the marginal distribution $\XX$ is
\begin{equation}
\mathbb P(X_1=x_1,...,X_J = x_J) = \int_{\mathbb R^K} \prod_{j=1}^J \mathbb  P( X_j =x_j| \ttt) \pi(\ttt) d \ttt.
\end{equation}
	In latent variable modeling, it is important to keep $K$, the dimension of the latent vector,  strictly less than $J$, that of the observed data. In fact, in most cases, $K$ is much smaller than  $J$.
	As mentioned previously, a low-dimensional latent variable model is often insufficient to capture all the dependence among $\XX$.
	We take the multidimensional item response model as the basic latent variable model and further add a graphical component to it.

	We consider the Ising model as the graphical component that is an undirected graphical model, also known as Markov random field. The specification of an undirected graphical model consists of a set of vertices $V = \{1,...,J\}$ and a set of edges $E \subset V \times V$. The graph is undirected in the sense that $(i,j)\in E$ if and only if $(j,i) \in E$.
	We associate a Bernoulli random  variable $X_j$ to each vertex $j\in V$. The graph encodes the conditional dependence structure among $X_1,...,X_J$.
	In particular, vertices $i$ and $j$ do not have an edge, $(i,j)\notin E$, if $X_i$ and $X_j$ are conditionally independent given all others, $\{X_l: l \neq i \mbox { or } j\}$.
	The Ising model parameterizes an undirected graph via the exponential family admitting the following probability mass function
\begin{equation}\label{ising}
\bar f(\xx) = \frac 1 {z(S)} \exp\Big\{\frac 1 2 \xx ^\top S \xx\Big\},
\end{equation}
where $S = (s_{ij})$ is a $K$ by $K$ symmetric matrix, i.e., $s_{ij} = s_{ji}$, and $z(S)$ is the normalizing constant
\begin{equation}\label{NC}
z(S) = \sum_{\xx \in \{0,1\}^K} \exp\Big\{\frac 1 2 \xx ^\top S \xx\Big\}.
\end{equation}
	The matrix $S$ maps to a graphical structure. There is an edge between vertices $i$ and $j$, $(i,j)\in E$, if and only if $s_{ij} = s_{ji} \neq 0$. According to the probability mass function \eqref{ising}, it is easy to check that $X_i$ and $X_j$ are conditionally independent given all other $X_l$'s, $l \neq i$ or $j$, if $s_{ij} = 0$.

\subsection{Fused latent and graphical model}

We propose a fused latent and graphical (FLaG) model that combines the IRT model and the Ising model. To do so, we present another representation of the IRT model. We write the item response function \eqref{IRF} as
$$\mathbb P(X_j = x_j|\ttt) = \frac{e^{(a_j^\top \ttt + b_j) x_j}}{1+ e^{a_j^\top \ttt + b_j}} \propto e^{(a_j^\top \ttt + b_j) x_j}.$$
With the local independence assumption, the joint conditional distribution is
$$\tilde f(\xx)  \propto \exp\Big \{ \sum _{j=1}^J (a_j^\top \ttt + b_j) x_j\Big\}
= \exp\Big \{ \ttt ^\top A^\top \xx + \bb^\top \xx\Big\}, $$
where $A = (a_1,...,a_J) = (a_{jk})_{J\times K}$ and $\bb = (b_1,...,b_J)^\top.$

\begin{remark}
Throughout this paper, we frequently use the notation ``$\propto$'' to define probability density or mass functions. It means that the left-hand side and the right-hand side are different by a factor that depends only on the parameters
and is free of the value of the random variable/vector. The constant can be obtained by summing or integrating out the random variable/vector. Such a constant sometimes could be difficult to obtain, which will be discussed in the sequel.
\end{remark}

With this representation, the probability mass function of the Ising model in \eqref{ising} can be similarly written as
$$\bar f(\xx) \propto  \exp\Big\{\frac 1 2 \xx ^\top S \xx\Big\}.$$
We combine these two models and write
\begin{equation}\label{LGM}
f(\xx |\ttt, A,S) \triangleq \mathbb P (X_1= x_1,...,X_J = x_J|\ttt, A, S) \propto \exp\Big \{ \ttt ^\top A^\top \xx + \frac 1 2 \xx ^\top S \xx\Big\}.
\end{equation}
We remove the term $\bb ^\top \xx$, because it is absorbed into the diagonal terms of $S$. Notice that $x_j\in \{0,1\}$ and thus $x_j = x_j^2$. The squared terms in \eqref{LGM} becomes linear  $\sum_{j=1}^J s_{jj}x_j^2 = \sum_{j=1}^J s_{jj}x_j$.
For technical convenience, we further impose a prior distribution on $\ttt$ such that the joint distribution of $(\XX, \ttt)$ given the parameters $(A,S)$ is
\begin{equation}
f(\xx, \ttt | A,S) \propto \exp\Big \{ - \frac 1 2 \Vert \ttt\Vert^2 + \ttt ^\top A^\top \xx + \frac 1 2 \xx ^\top S \xx\Big\},
\end{equation}
where $\Vert \cdot \Vert$ is the  usual Euclidean norm on $\mathbb R^K$.
Define the normalizing constant
$$z(A,S) = \sum_{\xx\in \{0,1\}^J}\int_{\mathbb R^K} \exp\Big \{- \frac 1 2 \Vert\ttt\Vert^2+ \ttt ^\top A^\top \xx + \frac 1 2 \xx ^\top S \xx\Big\} d\ttt.$$
The complete data likelihood function of a single observation is
\begin{equation}\label{JointLike}
f(\xx , \ttt | A,S) = \frac{1}{z(A,S)} \exp\Big \{ - \frac 1 2 \Vert \ttt\Vert^2 + \ttt ^\top A^\top  \xx + \frac 1 2 \xx ^\top S \xx\Big\}.
\end{equation}
	The normalizing constant $z(A,S)$ is not easy to compute and thus evaluation of the above likelihood is not straightforward. We will address this issue momentarily.

Both the IRT and the Ising models are special cases of \eqref{LGM}.
	By setting $a_{jk} = 0$, \eqref{LGM} recovers the Ising model with parameter matrix $S$; by setting $s_{ij}=0$ for $i\neq j$, \eqref{LGM} is equivalent to an IRT model.
	Conditional on $\ttt$, $\XX$ follows the Ising model, in particular,
	$$\mathbb  P(\XX = \xx |\ttt, A,S) \propto \exp\Big\{\frac 1 2 \xx^\top S(\ttt) \xx\Big\},$$
	where $s_{ij} (\ttt) = s_{ij}$ for $i \neq j$ and $s_{jj} (\ttt) = s_{jj} + 2 a_j^\top \ttt $. The graphical structure $S$, in particular, $\{s_{ij}: i\neq j\}$, captures the remaining dependence that is not explained by the latent vector.
	For each $X_j$, if we further condition on the rest of the random variables $\XX_{-j} = (X_i: i\neq j)$, the conditional distribution admits the form of a logistic model
\begin{equation*}
\mathbb P(X_j=1 | \XX_{-j} =\xx_{-j}, \ttt, A,S) = \frac{\exp\{\frac{s_{jj}(\ttt)}{2} + \sum_{i\neq j } s_{ij}x_i\}}{1+ \exp\{\frac{s_{jj}(\ttt)}{2} + \sum_{i\neq j } s_{ij}x_i\}}.
\end{equation*}
Thus, the conditional distribution can be written in a closed form, though the joint likelihood \eqref{JointLike} is often difficult to evaluate.

	Lastly, we consider the marginal joint distribution of $\XX$ with the latent vector $\ttt$ integrated out, more precisely,
\begin{equation}\label{ML}
f(\xx | A,S) = \int_{\mathbb R^K} f(\xx ,\ttt | A,S ) d\ttt = \frac{(2\pi)^{K/2}}{z(A,S)} \exp\Big\{\frac{1}{2} \mathbf \xx^{\top} (A A^\top + S) \mathbf \xx\Big \}.
\end{equation}
As the latent vector $\ttt$ is not directly observed, our subsequent analysis of the estimation is mostly based on the above marginal likelihood.
	Notice that the loading matrix $A$ enters the likelihood function $f(\xx|A,S)$ in the form of $A A^\top$. Therefore, $A$ is not identifiable by itself. We reparameterize the likelihood function and define $L= A A^\top$. With a  slight abuse of notation, we write
	$$f(\xx | L, S) = \frac{(2\pi)^{K/2}}{z(A,S)} \exp\Big\{\frac{1}{2} \mathbf \xx^{\top} (L + S) \mathbf \xx\Big \}.$$
	This is mostly because the latent vector $\ttt$ is not directly observed and its loading matrix $A$ can only be identified up to a non-degenerate  transformation.
Note also that there is an identifiability issue between $L$ and $S$, as the two matrices enter the marginal likelihood function in the form of $L+S$.
	In particular, $L$ characterizes the dependence among $\XX$ that is due to the latent structure and $S$ characterizes that of the graphical structure. In the analysis, assumptions will be imposed on the parameter space so that $L$ and $S$ are separable from each other based on the data.

\section{On maximum regularized pseudo-likelihood estimator}\label{sec:est}

\subsection{Estimation}

	In this section, we address issues related to estimation of the latent  graphical model described in the previous section including evaluation of the likelihood function, dimension estimation/reduction of the latent vector, estimation of the conditional graph, parameter identifiability, and oracle property of the proposed estimator.
	To begin with, we assume that all parameters including the dimension of the latent vector $\ttt$ and the conditional graph are unknown.

	The first issue concerning the estimation is that evaluation of the marginal likelihood function \eqref{ML} involves the normalizing constant $z(A,S)$ whose computational complexity grows exponentially fast in the dimension $J$. In fact, its computation is practically infeasible even for a moderately large $J$.
	We take a slightly different approach by considering the conditional likelihood of $X_j$ given $\XX_{-j}$, which has a closed form as discussed previously. Let $L= (l_{ij})$ and $S= (s_{ij})$. We have
\begin{equation}\label{eq:logis}
\mathbb P(X_j=1 | \XX_{-j} =\xx_{-j},L,S) = \frac{\exp\{\frac{l_{jj}+s_{jj}}{2} + \sum_{i\neq j } (l_{ij}+s_{ij})x_i\}}{1+ \exp\{\frac{l_{jj}+s_{jj}}{2} + \sum_{i\neq j } (l_{ij}+s_{ij})x_i\}}.
\end{equation}
This closed form is crucial for our inference. Let
\begin{equation}\label{eq:con_lik}
\mathcal L_j(L, S; \xx) \triangleq \mathbb P(X_j = x_j | \XX_{-j} =\xx_{-j},L,S)
\end{equation}
denote the conditional likelihood for $X_j$ given $\XX_{-j}$. Our estimation is based on a pseudo-likelihood function by multiplying all the conditional likelihood together.
The pseudo-likelihood based on $N$ independent observations is
\begin{equation}\label{PL}
\mathcal L(L,S) = \prod_{i=1}^N \prod_{j=1}^J \mathcal L_j (L,S; \xx_i),
\end{equation}
where $\xx_i$ is the $i$th observation.
	
	In the above pseudo-likelihood, $L$ and $S$ are unknown parameters. Besides, the dimension of the latent vector $\ttt$ and the conditional graphical structure implied by $S$ are also unknown.
	 We will estimate the set of edges $E$.
	As for the dimension of $\ttt$, to ensure identifiability, we assume that the loading matrix $A$ is of full column rank; otherwise, we can always reduce the dimension $K$ and make $A$ full column rank. Thus, $L = A A^\top$ also has rank $K$. Notice that $L$ is a positive semidefinite matrix. The rank of $L$ is the same as the number of its non-zero eigenvalues. To estimate the conditional graph and the dimension of the latent vector, we impose regularization on  $S$ and $L$.
	
	As mentioned previously, the parameters $L$ and $S$ enter the likelihood function in the form of $L+S$. In principle, one cannot identify $L$ from $S$ based on the data only. We will impose additional assumptions to ensure their identifiability (or uniqueness of the estimator) based on the following rationale.
	We believe that the multidimensional IRT model (with the local independence assumption) is largely correct. The latent vector accounts for most dependence/variation of the multivariate response vector \XX. In the context of cognitive assessment, this is interpreted as that a person's responses to items are mostly driven by a few latent attributes.
	The remaining dependence is rather low. Thus, a crucial assumption in our estimation is that the graphical structure explains a small portion of the dependence in \XX.
	To quantify this assumption, we assume that the matrix $S$ is sparse. In addition, the dimension of the latent vector stays low.
	These assumptions will be made precise in later discussions where theoretical properties of our estimator are established.

	Based on the above discussion, we propose an estimator by optimizing a regularized pseudo-likelihood
\begin{equation}\label{est}
(\hat L, \hat S) = \arg\min_{L,S}\Big\{ - \frac 1 N\log \{ \mathcal L(L,S)\} + \gamma \Vert \OO(S)\Vert_1 + \delta \Vert L\Vert_*\Big \}
\end{equation}
where $\mathcal L(L,S)$ is defined by \eqref{PL} and the minimization is subject to the constraints that $L$ is positive semidefinite and $S$ is symmetric.  Throughout this paper, we use $L \succeq 0$ to denote that $L$ is positive semidefinite.

	We provide some explanations of the two penalty terms $\Vert \OO(S)\Vert_1$ and $\Vert L\Vert_*$.
	In the first term, $\OO(S)$ is a $J\times J$ matrix such that it is identical to $S$ except that its diagonal entries are all zero, that is, $\OO(S)=(\tilde s_{ij}) $ where $\tilde s_{ij} = s_{ij}$ for $i\neq j$ and $\tilde s_{ii}=0$.  Thus,
$$\Vert \OO(S)\Vert_1 = \sum_{i\neq j} |s_{ij}|,$$
which penalizes the number of nonzero $s_{ij}$'s that is also the number of edges in the conditional Ising model.
	Notice that we do not penalize the diagonal elements of $S$ because $s_{jj}$ controls the marginal distribution of  $X_j$. As mentioned previously, the constant term $b_j$ in  the IRT model is  absorbed into the diagonal term $s_{jj}$.
	By increasing the regularization parameter $\gamma$, the number of nonzero off-diagonal elements decreases and thus the number of edges in the conditional graph also decreases.
	The $L_1$ penalty was originally proposed in \cite{tibshirani1996regression} for linear models and later in the context of graphical models \citep{meinshausen2006high,friedman2008sparse,hofling2009estimation,ravikumar2010high,guo2010joint}.


	The second penalty term is  $\Vert L \Vert_* = \mathrm{Trace} (L)$. Notice that $L = A^\top A$ is a positive semidefinite matrix and admits the following eigendecomposition
	$$L = T^\top \Lambda T,$$
	where $T$ is an orthogonal matrix,  $\Lambda =\mathrm{diag}\{\lambda_1,..., \lambda_J \}$ and $\lambda_j \geq 0$. The nuclear norm can be alternatively written as
	$$\Vert L \Vert _* = \sum_{j=1}^J |\lambda_j|.$$
	Therefore, $\Vert L \Vert _*$ penalizes the number of nonzero eigenvalues of $L$, which is the same as the rank of $L$.
	This regularization is first proposed in \cite{fazel2001rank} and its statistical properties are studied in \cite{bach2008consistency}.
	The estimators $\hat L$ and $\hat S$ depend on the regularization parameters $\gamma$ and $\delta$, whose choice will be described in the sequel. To simplify notation, we omit the indices $\gamma$ and $\delta$ in the notation $\hat L$ and $\hat S$.

	The regularized estimators $\hat L$ and $\hat S$ naturally yield estimators of the dimension of $\ttt$ and the conditional graph $E$. In particular, an estimator of the dimension of $\ttt$ is
\begin{equation} \label{EstK}
\hat K = \mathrm{rank}(\hat L)
\end{equation}
and an estimator of the conditional graph is
\begin{equation}\label{EstE}
\hat E = \{(i,j):\hat s_{ij} \neq 0\}.
\end{equation}
In what follows, we state the theoretical properties of this  regularized pseudo-likelihood estimator.

\subsection{Theoretical properties of the estimator}

In this subsection, we present the properties of the regularized estimator $(\hat L, \hat S)$ defined as in \eqref{est} and the estimators $\hat K$ and $\hat E$ defined as in \eqref{EstK} and \eqref{EstE}.
Throughout the discussion, let $L^*$ and $S^*$ denote the true model parameters.

To state the assumptions, we first need the following technical developments.
The pseudo-likelihood (and the likelihood) function depends on $L$ and $S$ through $L+S$. Define
 \begin{equation}\label{eq:def:h_fun}
h_N(L+S) =  -\frac 1 N \log \{\mathcal L(L,S)\}.
\end{equation}
If we reparameterize $M= L+S$, its information associated with the pseudo-likelihood is given by
\begin{equation}\label{info}
\II=\E\Big\{\frac{\partial^2 h_1}{\partial^2 M}\Big|_{M=M^*}\Big \},
\end{equation}
which is a $J^2$ by $J^2$ matrix and $M^*= L^* + S^*$ is the true parameter matrix.
For a differentiable manifold $\mathcal M$, we let $T_x \mathcal M$ denote its tangent space at $x\in \mathcal M$. We refer to \cite{sternberg1965lectures} for the definition of a manifold and its tangent space.
The first condition, which ensures local identifiability, is as follows.
\begin{itemize}
\item[A1] The matrix $\mathcal I^*$ is positive definite restricted to the set
$$\mathcal M \triangleq \{M = L+S: \mbox{$L$ is positive semidefinite and $S$ is symmetric}\}.$$
That is, for each vector $v\in  \mathcal M$, $v^\top \II v \geq 0$ and the equality holds if and only if $v=0$.
\end{itemize}
In what follows, we describe a few submanifolds of $\mathbb R^{J\times J}$  and their tangent spaces.
Let $\mathcal S^*$ be the set of symmetric matrices admitting the same sparsity as that of $S^*$, that is,
\begin{equation*}
\mathcal S^* = \{S: \mbox{$S$ is a $J\times J$ symmetric matrix and $s_{ij} = 0$ if $s^*_{ij} =0$ for all $i\neq j$}\}.
\end{equation*}
On considering that $\mathcal S^*$ is a submanifold of $\mathbb R^{J\times J}$,  its tangent space at $S^*$ is $\mathcal S^*$ itself, that is,
$$T_{S^*} \mathcal S^* = \mathcal S^*.$$
Define the set of matrices
$$
\mathfrak{L}=\{L: L \mbox{ is positive semidefinite and } \mathrm{rank}(L)\leq K\},
$$
where $K=\mathrm{rank}(L^*)$. The set $\mathfrak{L}$ is differentiable in a neighborhood of $L^*$.  Therefore,
it is a submanifold of $\mathbb R^{J\times J}$ within the neighborhood of $L^*$ and its tangent space at $L^*$ is well defined.
	To describe the tangent space of $\mathfrak L$ at $L^*$, we consider its eigendecomposition
	$$L^*= U_1^*D_1^*U_1^{*\top},$$
	where $U_1^*$ is a $J\times K$ matrix satisfying $U_1^{*\top} U_1^*=I_K$, $I_K$ is the $K\times K$ identity matrix, and $D_1^*$ is a $K\times K$ diagonal matrix consisting of the (positive) eigenvalues of $L^*$. Then, the tangent space of $\mathfrak L $ at $L^*$ is
	$$T_{L^*}\mathfrak L  = \{U^*_1 Y+Y^{\top} U_1^{*\top}: Y \mbox{ is a $K\times J$ matrix} \}.$$
We make the following assumptions on $L^*$ and $\MS$.
\begin{enumerate}
	\item[A2] The positive eigenvalues of $L^*$ are distinct.
	\item[A3] The intersection between $\MS$ and $T_{L^*}\mathfrak L$ is trivial, that is, $\MS\cap T_{L^*}\mathfrak L= \{\mathbf{0}_{J\times J} \}$, where $\mathbf{0}_{J\times J}$ is the $J\times J$ zero-matrix. 
\end{enumerate}
Lastly, we present an irrepresentable condition that is key to the consistency of the regularized estimator.
	Define a linear operator $\FF: \MS\times T_{L^*}\mathfrak L \to  \MS\times T_{L^*}\mathfrak L$,
	\begin{equation}\label{F}
		\FF(S,L) = (\PP_{\MS}\{ \II(S+L)\}, \PP_{T_{L^*}\mathfrak L}\{ \II(S+L)\}),
	\end{equation}
where $\II$ is the $J^2 \times J^2$ matrix in \eqref{info}. With a slight abuse of notation,  we let $\II(S+L)$ denote matrix-vector multiplication where $S$ and $L$ are vectorized with their elements being arranged in the same order as the order of the derivatives in $\II$.
	The map $\PP_{\mathcal M} (A)$ is the projection operator of matrix $A$ onto the manifold $\mathcal M$ with respect to  the inner product for matrices,
	$$A\cdot B=\sum_{i=1}^J\sum_{j=1}^{J} A_{ij}B_{ij}=\mathrm{Trace}(AB^\top).$$
	That is, $\PP_{\mathcal M} (A)$ is the matrix in $\mathcal M$ minimizing the distance to $A$ induced by the matrix inner product ``$\cdot$''.
We define a linear operator $\FF^{\bot}:\MS\times \ML^*\to \MS^{\bot}\times \ML^{*\bot}$,
$$
\FF^{\bot}(S',L')= (\PP_{\MS^{\bot}}\{\II(S'+L')\}, \PP_{(T_{L^*}\mathfrak L)^{\bot}}\{\II(S'+L')\}).
$$
For a linear subspace $\mathcal M$, $\mathcal M^\bot$ denote its orthogonal complement in $\mathbb R^{J\times J}$.
For a matrix $A= (a_{ij})$, we apply the $\mathrm{sign}$ function  to each of its element, that is
$$\mathrm{sign}(A) = (\mathrm{sign}(a_{ij})).$$
Furthermore, for each constant $\rho>0$, define a norm for a matrix couple $(A,B)$ of appropriate dimensions such that
	$$
	\|(A,B)\|_{\rho}=\max(\|A\|_{\infty}, \|B\|_{2}/\rho ),
	$$
where $\|\cdot\|_{\infty}$ and $\|\cdot \|_{2}$ are the maximum and spectral norm respectively. Here, the spectral norm $\|B\|_2$ is defined as the largest eigenvalue of $B$ for a positive semidefinite matrix $B$.
The last condition is stated as follow
\begin{itemize}
	\item[A4] There exists a positive constant $\rho$ such that
	\begin{equation}\label{eq:irre}
	\Big\|\FF^{\bot} \FF^{-1}(\mathrm{sign}(\OO(S^*)), \rho U_1^* U_1^{*\top})\Big\|_{\rho}<1.
	\end{equation}
\end{itemize}
The following lemma guarantees that $\FF^{-1}$ in \eqref{eq:irre} is well defined.
\begin{lemma}\label{lemma:f-inv}
	Under Assumptions A1 and A3, the linear operator $\FF$ is invertible over $\MS\times T_{L^*}\mathfrak L$.
\end{lemma}
With these conditions, we present the theoretical properties of our estimator.

\begin{theorem}\label{thm:consistent}
Under Assumptions A1-A4,
choose the tuning parameter $\delta_N=\rho\gamma_N=N^{-\frac{1}{2}+\eta}$
for some sufficiently small positive constant $\eta$, and  $\rho$ satisfying \eqref{eq:irre}. Then, the optimization  \eqref{est} has a unique  solution $(\hat S, \hat L)$ that converges in probability to the true parameter
$(S^*,L^*)$. In addition, $(\hat S, \hat L)$ recovers the sparse and low rank structure of $(S^*,L^*)$ with probability tending to $1$, that is,
$$\lim_{N\to\infty}\mathbb{P}\Big\{\mathrm{sign}(\hat S)= \mathrm{sign}(S^*), \mathrm{rank}(\hat L)=\mathrm{rank}(L^*)\Big\}=1.$$
\end{theorem}

We provide a discussion on the technical conditions. Condition A1 ensures local identifiability of the parameter $M= S+L$. Given the likelihood function is log-concave, the parameter $M$ can be estimated consistently by the pseudo-likelihood.
Condition A3 corresponds to the \emph{transversality condition} in \cite{chandrasekaran2012latent}. Lastly, Condition A4 is similar to the irrepresentable condition \citep{zhao2006model,jia2008model} that plays an important role in the
consistency of sparse model selection based on $L_1$-norm regularization.


\subsection{On the choice of tuning parameters}\label{sec:bic}

	Theorem \ref{thm:consistent} provides a theoretical  guideline of choosing the regularization parameters $\gamma$ and $\delta$. Nonetheless, it leaves quite some freedom. In what follows, we provide a more specific choice of $\gamma$ and $\delta$ that will be used in  the simulation study and the real data analysis.

We consider to choose $\gamma$ and $\delta$ to
minimize the Bayes information criterion \citep[BIC;][]{Sch78},
that is known to yield consistent variable selection. BIC is defined as
$$\text{BIC}(\mathcal{M}) = -2 \log L_N(\hat \beta(\mathcal{M})) + \vert \mathcal{M}\vert \log N,$$
where $\mathcal{M}$ is the current model, $L_N(\hat \beta(\mathcal{M}))$ is the maximal likelihood for a given  model $\mathcal{M}$, and $\vert \mathcal{M}\vert$ is the number of free parameters in $\mathcal{M}$.
In this study, we replace the likelihood function with the pseudo-likelihood function. To avoid ambiguity, we change the notation and use $\hat L^{\gamma, \delta}$ and $\hat S^{\gamma, \delta}$ to denote the estimator in \eqref{est} corresponding to regularization parameters $\gamma$ and $\delta$. Let
\begin{align*}
\mathcal{M}^{\gamma, \delta} = \left\{\right. & L \mbox{ is positive semidefinite and } S \mbox{ is symmetric}, \\
&\text{rank}(L) \leq \text{rank}(\hat L^{\gamma, \delta}) \mbox{ and } s_{ij} = 0 \mbox{ if } \hat s_{ij}^{\gamma, \delta} = 0 \mbox{ for all } i\neq j\left. \right\}
\end{align*}
be the submodel selected by the tuning parameters $(\gamma, \delta)$. It contains all models in which the positive semidefinite matrix $L$ has rank no larger than that of $\hat L^{\gamma,\delta}$ and the symmetric matrix $S$ has
the same support as $\hat S^{\gamma,\delta}$.
We select the tuning parameters $\gamma$ and $\delta$ such that the corresponding model minimizes
the Bayesian information criterion based on the pseudo-likelihood
\begin{equation}\label{BIC_Chap2}
\text{BIC}(\mathcal{M^{\gamma, \delta}}) = -2 \max_{(L,S)\in M^{\gamma, \delta}}\{\log \mathcal L(L,S) \} + \vert \mathcal{M^{\gamma,\delta}}\vert \log N,
\end{equation}
where the number of parameters in $\mathcal{M^{\gamma, \delta}}$ is
$$\vert \mathcal{M^{\gamma, \delta}}\vert = \left(JK - \frac{(K-1)K}{2}\right) + \sum_{i\leq j} 1_{\{\hat s^{\gamma, \delta}_{ij} \neq 0\}},$$
\mbox{for $\text{rank}(\hat L^{\gamma, \delta}) = K$} .
The two terms are the numbers of free parameters in $L$ and $S$ respectively. Specifically, the number of free parameters in $L$ is counted as follows. Let $L = U_1 D_1 U_1^\top$ be the eigendecomposition of $L$, where $D_1$ is a $K\times K$ diagonal matrix and columns of $U_1$ are unit-length eigenvectors of $L$. $D_1$ has $K$ parameters and $U_1$ has $JK - {K(K+1)}/{2}$ parameters due to constraint $U_1^\top U_1 = I_K$.
Combining them together, $L$ has $JK - {(K-1)K}/{2}$ parameters.

Maximizing the pseudo-likelihood in \eqref{BIC_Chap2} is no longer a convex optimization problem. However, our experience shows that this nonconvex optimization can be solved stably using
a generic numerical solver, with starting point $(\hat L^{\gamma, \delta}, \hat S^{\gamma, \delta})$.
The tuning parameters are finally selected by
$$(\hat \gamma, \hat \delta) = \arg\min_{\gamma, \delta} \text{BIC}(\mathcal{M^{\gamma, \delta}}).$$
In addition, the corresponding maximal pseudo-likelihood estimates of $L$ and $S$ are used as the final estimate of $L$ and $S$:
\begin{equation}\label{MLE_final}
(\hat L, \hat S) = \arg\max_{(L,S)\in M^{\hat \delta, \hat\gamma}}\{\mathcal L(L,S) \}.
\end{equation}

\section{Computation}\label{sec:comp}

In this section, we describe the computation of the regularized estimator in \eqref{est}, which is not straightforward for two reasons.  First, the coordinate-wise descent algorithms \citep{fu1998penalized,friedman2007pathwise}, which are widely used in convex optimization problems with $L_1$ norm regularization,
do not apply well to this problem.
These  algorithms optimize the objective function with respect to one parameter at a time.
For our case, updating with respect to $s_{ij}$ is not in a closed form.
Moreover, the optimization is constrained on a space where the matrix $L$ is positive semidefinite.
As a consequence, it becomes a
semidefinite programming problem, for which a standard approach is the interior point methods \citep[e.g.][]{boyd2004convex}.
The computational cost for each iteration and the memory requirements
of an interior
point method are prohibitively high for this problem, especially when
$J$ is  large.

We propose a  method that avoids these problems by taking advantage of the special structure of the $L_1$ and nuclear norms by means of the alternating direction method of multiplier \citep[ADMM;][]{boyd2011distributed, glowinski1975approximation, gabay1976dual}.
The key idea is to decompose the optimization of \eqref{est} into subproblems that can be solved efficiently.

Consider two closed convex functions $$f: \mathcal{X}_f  \rightarrow \mathbb R \mbox{~~~~and~~~~} g: \mathcal{X}_g  \rightarrow \mathbb R,$$
where the domains $\mathcal{X}_f$ and $\mathcal{X}_g$ of functions $f$ and $g$ are closed convex subsets of $\mathbb R^n$, and $\mathcal{X}_f \cap \mathcal{X}_g$ is nonempty. Both $f$ and $g$ are possibly  nondifferentiable.
The alternating direction method of multiplier is an iterative algorithm that solves the following generic optimization problem:
\begin{equation*}
\begin{aligned}
\min_{x \in \mathcal{X}_f \cap \mathcal{X}_g} &\ \{f(x) + g(x)\},\\
\end{aligned}
\end{equation*}
or equivalently
\begin{equation}\label{ADMM}
\begin{aligned}
\min_{x \in \mathcal{X}_f, z \in \mathcal{X}_g} &\ \{f(x) + g(z)\}.\\
\mbox{ s.t. }  &\ \ x = z
\end{aligned}
\end{equation}
To describe the algorithm, we first define proximal operators $\text{\bf P}_{\lambda, f}$: $\mathbb R^n \rightarrow \mathcal{X}_f$ as
$$\text{\bf P}_{\lambda, f}(v) = \arg\min_{x \in \mathcal{X}_f} \{ f(x) + \frac{1}{2\lambda} \Vert x - v \Vert^2\} $$
and
$\text{\bf P}_{\lambda, g}$: $\mathbb R^n \rightarrow \mathcal{X}_g$
$$\text{\bf P}_{\lambda, g}(v) = \arg\min_{x \in \mathcal{X}_g} \{ g(x) + \frac{1}{2\lambda} \Vert x - v \Vert^2\}, $$
where $\Vert \cdot\Vert$ is the usual Euclidean norm on $\mathbb R^n$ and $\lambda$ is a scale parameter that is a fixed positive constant.
The algorithm starts with some initial values $x^0 \in \mathcal{X}_f$, $z^0 \in \mathcal{X}_g$, $u^0\in \mathbb R^n$. At the $(m+1)$th iteration, $(x^m, z^m, u^m)$ is updated according to the following steps until convergence
\begin{itemize}[leftmargin=2cm]
\item [Step 1: ] $x^{m+1} := \text{\bf P}_{\lambda, f}(z^m - u^m)$;
\item [Step 2: ] $z^{m+1} := \text{\bf P}_{\lambda, g}(x^{m+1} + u^m)$;
\item [Step 3: ] $u^{m+1} := u^{m} + x^{m+1} - z^{m+1}$.
\end{itemize}
The algorithm is fast when the proximal operators $\text{\bf P}_{\lambda, f}$ and $\text{\bf P}_{\lambda, g}$
can be efficiently evaluated.
The convergence properties of the algorithm are summarized in the following result in \cite{boyd2011distributed}. Let $p^*$ be the minimized value in \eqref{ADMM}.
\begin{proposition}[\citealp{boyd2011distributed}]
Assume functions $f$: $\mathcal{X}_f  \rightarrow \mathbb R$ and $g$: $\mathcal{X}_g  \rightarrow \mathbb R$ are closed convex functions, whose domains $\mathcal{X}_f$ and $\mathcal{X}_g$ are closed convex subsets of $\mathbb R^n$ and $\mathcal{X}_f \cap \mathcal{X}_g \neq \emptyset$.  Assume the Lagrangian of \eqref{ADMM}
$$L(x, z, y) = f(x) + g(z) + y^{\top}(x-z)$$
has a saddle point, that is, there exists $(x^*,z^*,y^*)$ (not necessarily unique) that $x^* \in \mathcal{X}_f$ and $z^* \in \mathcal{X}_g$, for which
$$L(x^*,z^*,y) \leq L(x^*,z^*,y^*) \leq L(x,z,y^*), \qquad \forall ~x, z, y \in \mathbb{R}^n.$$
Then the ADMM has the following convergence properties.
\begin{enumerate}
  \item Residual convergence. $x^m - z^m \rightarrow 0$ as $m \rightarrow \infty$; i.e., the iterates approach feasibility.
  \item Objective convergence. $f(x^m) + g(z^m) \rightarrow p^*$ as $ m \rightarrow \infty$; i.e., the objective function of the iterates approaches the optimal value.
\end{enumerate}
\end{proposition}
We would like to point out that the assumption on the Lagrangian $L(x, z, y)$ is mild \citep[see Chapter 5,][]{boyd2004convex}. In particular, if
strong duality holds for the problem \eqref{ADMM} and let $(x^*, z^*)$ and $y^*$ be the corresponding
primal and dual optimal points, $(x^*,z^*,y^*)$ forms a saddle point for the Lagrangian.

We now adapt this algorithm to the optimization of the regularized pseudo-likelihood. In particular, we reparameterize $M = L+S$ and let $x = (M, L, S)$ (viewed as a vector). Let $h_N$ be defined as in \eqref{eq:def:h_fun}. We define
\begin{enumerate}[leftmargin=0cm]
\item[] $\mathcal{X}_f = \{(M, L, S): M, L, S \mbox{ are } J \times J \mbox{ matrices },  L \mbox{ is positive semidefinite, } S \mbox{ is symmetric}\},$
\item[] $f(x) =  h_N(M) + \gamma \Vert \OO(S)\Vert_1 + \delta \Vert L\Vert_*,$ $x \in \mathcal{X}_f$,
\item[] $\mathcal{X}_g = \{(M, L, S): M, L, S \mbox{ are } J \times J \mbox{ matrices }, M \mbox{ is symmetric and } M = L + S \}$,
\item[] and $g(x) = 0$, $x \in \mathcal{X}_g$.
\end{enumerate}
Obviously, the optimization \eqref{est} can be written as $$\min_{x \in \mathcal{X}_f \cap \mathcal{X}_g}~\{ f(x) + g(x)\}.$$
In addition, it is easy to verify that $\mathcal{X}_g$ is a closed convex set and $g$
is a closed convex function. Furthermore, $\mathcal{X}_f$ is closed and it is also convex since the symmetric and positive semidefinite constraints are convex constraints.
$h_N(M)$ is convex, since the pseudo-likelihood function is the sum of several log-likelihood functions of the logistic models that are all concave \citep[See Chapter 7,][]{boyd2004convex}. Because the $L_1$ and nuclear norms are convex functions, $f(x)$ is a convex function.
	Thanks to the continuity, $f$ is closed. In summary, $f$ is a closed convex function on its domain $\mathcal{X}_f$.

We now present each of the three steps of the ADMM algorithm and show that the proximal operators $\text{\bf P}_{\lambda, f}$ and $\text{\bf P}_{\lambda, g}$ are easy to evaluate.
Let
	$$x^m = (M^m, L^m, S^m), \quad z^m = (\tilde M^m, \tilde L^m, \tilde S^m),\quad  u^m = (U^m_M, U^m_L, U^m_S).$$
	
	Step 1. We solve $x^{m+1} = \text{\bf P}_{\lambda, f}(z^m - u^m)$. Due to the special structure of $f(\cdot)$, $M^{m+1}, L^{m+1}$, and $S^{m+1}$ can be updated separately. More precisely,

\begin{equation}\label{update1}
M^{m+1} = \arg\min_{M}~~ h_N(M) + \frac{1}{2\lambda} \Vert M - (\tilde M^m - U_{M}^m) \Vert_F^2;
\end{equation}

\begin{equation}\label{update2}
\begin{aligned}
L^{m+1} &= \argmin_{L}~~\delta \Vert L \Vert_{*} + \frac{1}{2\lambda} \Vert L - (\tilde L^m - U_{L}^m) \Vert_F^2,~~~~\\
\mbox{ s.t. } & \mbox{ $L$ is positive semidefinite};
\end{aligned}
\end{equation}

\begin{equation}\label{update3}
\begin{aligned}
S^{m+1} &= \argmin_{S}~~ \gamma \Vert \OO(S)\Vert_1 + \frac{1}{2\lambda} \Vert S - (\tilde S^m - U_{S}^m) \Vert_F^2,\\
\mbox{ s.t. } & \mbox{ $S$ is symmetric,}
\end{aligned}
\end{equation}
where $\Vert \cdot \Vert_F$ is the matrix Frobenius norm, defined as $\Vert M \Vert_F^2 = \sum_{i,j} m_{ij}^2$ for a matrix $M =(m_{ij})$.
We now discuss the optimization problems \eqref{update1}-\eqref{update3}. First, \eqref{update2} and \eqref{update3} can be computed in closed forms. More precisely, when $\tilde L^m - U_{L}^m$ and $\tilde S^m - U_{S}^m$ are both symmetric matrices (which is guaranteed when $\tilde M^0$, $\tilde L^0$, $\tilde S^0$, $U_M^0$, $U_L^0$, and $U_S^0$ are chosen to be symmetric),
$$ L^{m+1}=  T \text{diag}(\Lambda - \lambda\delta)_+ T^{\top},$$
where $\tilde L^{m} - U_L^m=T \Lambda T^{\top}$ is its eigendecomposition and $\text{diag}(\Lambda - \lambda\delta)_+$ is a diagonal matrix with its $j$th diagonal element being $( \Lambda_{jj} - \lambda\delta)_+$. The
operation $( \Lambda_{jj} - \lambda\delta)_+$ is called eigenvalue thresholding. 
In addition, $ S = ( s_{ij})$ is updated as
$$ s^{m+1}_{jj} = (\tilde S^{m} - U_S^m)_{jj}$$
and its off-diagonal entries are 
\[  s^{m+1}_{ij} = \left\{ \begin{array}{ll}
       (\tilde S^{m} - U_S^m)_{ij} - \gamma\lambda & \mbox{if\ }  (\tilde S^{m} - U_S^m)_{ij} > \gamma\lambda;\\
       (\tilde S^{m} - U_S^m)_{ij} + \gamma\lambda & \mbox{if\ }  (\tilde S^{m} - U_S^m)_{ij} < - \gamma\lambda;\\
        0 & \mbox{otherwise}.\end{array} \right. \]
Furthermore, solving \eqref{update1} is equivalent to solving $J$ $J$-dimensional unconstrained convex optimization problems. To see this, we denote $$M_j = (m_{1j}, ..., m_{Jj})^\top$$
as the $j$th column of a $J\times J$ matrix $M$. According to equation \eqref{eq:logis}, the conditional likelihood
$\mathcal L_j (L,S; \xx_i)$ defined by \eqref{eq:con_lik} can be written as a function of $M = L+S$ that only depends on $M_j$ and we denote it as $\mathcal L_j (M_j; \xx_i)$. As a result, evaluating \eqref{update1} can be decomposed into solving
$$\min_{M_j} -\frac{1}{N} \sum_{i=1}^N \mathcal L_j (M_j; \xx_i) + \frac{1}{2\lambda} \Vert M_j - (\tilde M^m - U_M^m)_j \Vert^2,$$
for $j = 1, 2,..., J$.
It can be solved efficiently using a standard solver, such as the Broyden-Fletcher-Goldfarb-Shanno method \citep[see e.g.][]{gentle2009computational}, where $J$ could be as large as a few hundreds.



Step 2. We solve $z^{m+1} = \text{\bf P}_{\lambda, g}(x^{m+1} + u^m)$.  Denote $\bar M = M^{m+1} + U_M^{m},\ \bar L  = L^{m+1} + U_L^{m}, \mbox{ and } \bar S  = S^{m+1} + U_S^{m}.$ Then evaluating $\text{\bf P}_{\lambda, g}(x^{m+1} + u^m)$ becomes:
\begin{align*}
  \min_{M, L, S}\ \  &  \frac{1}{2} \Vert M - \bar M \Vert^2_F + \frac{1}{2} \Vert L - \bar L \Vert ^2_F + \frac{1}{2} \Vert S - \bar S \Vert ^2_F,\\
\mbox{ s.t. }  &\mbox{$M$ is symmetric and $ M = L+S$}.
\end{align*}
This is a quadratic programming problem subject to linear constraints and thus can be solved in a closed form. Specifically,
\begin{equation*}
\begin{aligned}
\tilde M^{m+1} &= \frac{1}{3} \bar M + \frac{1}{3}  \bar M^\top  + \frac{1}{3} \bar L + \frac{1}{3} \bar S ,\\
\tilde L^{m+1} &= \frac{2}{3} \bar L + \frac{1}{6}  \bar M + \frac{1}{6}  \bar M^\top - \frac{1}{3}  \bar S,\\
\tilde S^{m+1} &= \frac{2}{3} \bar S + \frac{1}{6}  \bar M + \frac{1}{6}  \bar M^\top - \frac{1}{3} \bar L.
\end{aligned}
\end{equation*}


	Step 3 is a simple arithmetic. The advantage of the proposed algorithm is its low computational and memory cost   at each iteration.
In particular, the nondifferentiable $L_1$ and nuclear norms and the positive semidefinite constraint that induce difficulty in a generic solver are efficiently handled by closed-form updates.
In addition, the $J^2$-dimensional function $h_N(M)$
is decomposed to a sum of $J$  functions that can be optimized in parallel.

%

\section{Simulation Study and Real Data Analysis} \label{section: simulation}
In this section, we first conduct simulation studies to investigate the performance of the proposed methods.
Then we illustrate the method by analyzing a real data set of personality assessment.

\subsection{Simulation}

We consider $J = 30$ items and sample sizes $N = 250, 500, 1000, 2000,$ and $4000$ under the following three settings.
\begin{enumerate}
\item $K=1$ latent variable. For the $S$-matrix, all off-diagonal elements are zero except for  $s_{j,j+1} $  for $j = 1, 3, ..., 29$. There are in total 15 edges in the graph. This graph is equivalent to grouping the variables in pairs, \{1,2\}, \{3,4\}, ..., and \{29, 30\}. There is an edge between each pair.

\item $K=1$ latent variable. For $j= 1, 4, ..., 28$ , $s_{j,j+1}$, $s_{j,j+2}$, and $s_{j+1, j+2}$ are nonzero. There are 30 edges in the conditional graph. This is equivalent to grouping the variables in triples, \{1,2,3\}, \{4,5,6\}, ..., \{28,29,30\}. There are edges within the triple.

\item  $K=2$, and the conditional graph is the same as that of setting 1.
\end{enumerate}
The conditional graphs are visualized in Figure~\ref{fig:models}, where the upper and the lower panels represent the graphs $S$ in settings 1 and 2, respectively.
\begin{figure}[ht]
\centering
\includegraphics[scale = 0.3]{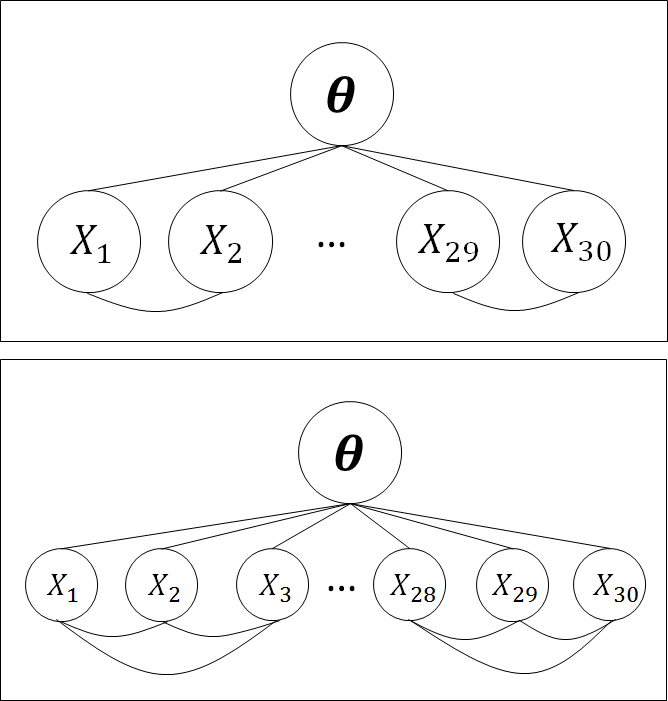}
\caption{The graphical representation corresponds to the two sparse network structures.}
\label{fig:models}
\end{figure}

For each model setting and each  sample size, we generate  50 independent data sets.
The tuning parameters are chosen based on  the  Bayesian information criterion as described in Section~\ref{sec:bic}.

\paragraph{Data generation.}

To generate a sample from the latent graphical model, we first generate $\ttt$ from its marginal distribution
$$f(\ttt) \propto \sum_{\xx\in \{0,1\}^J} \exp\left\{-\frac{1}{2}\Vert \ttt\Vert^2 + \xx^{\top} A\ttt +  \frac{1}{2} \xx^{\top} S \xx\right\}.$$
The above summation is computationally feasible because of the sparse graphical structure as in Figure \ref{fig:models}.
The latent vector $\ttt$ is sampled from the above marginal distribution by the accept/reject algorithm.
The conditional distribution of $\xx$ given $\ttt$ are independent between pairs and triples.

\paragraph{Evaluation criteria.}
To assess the performance of the dimension reduction and the estimation of the graph, we consider the criterion $C_1$. For a particular data set,
$C_1 = 1$ if and only if there exists a pair of $(\gamma, \delta)$, such that $\mathrm{rank}(\hat L^{\gamma, \delta})= \mathrm{rank}(L^*)$ and graph induced by  $\hat S^{\gamma, \delta}$ is the same as that by $S^*$, where $L^*$ and $S^*$ are the true parameters. 

Furthermore, we evaluate the BIC-based tuning parameter selection via
criteria $C_2$, $C_3$, and $C_4$. Let $(\hat L, \hat S)$ be the final estimates of the selected model defined as in \eqref{MLE_final}. Criterion $C_2$ evaluates the estimation of the rank of $L$,
$$C_2 = 1_{\{\text{rank}(\hat L) = \text{rank}(L^*)\}}.$$
In addition, $C_3$ evaluates the positive selection rate of the network structure of $S$, defined as
$$C_3 = \frac{\vert\{(i,j): i<j, \hat s_{ij}\neq 0, \mbox{ and } s^*_{ij}\neq 0\}\vert}{\vert\{(i,j): i<j, s^*_{ij}\neq 0\}\vert}.$$
Furthermore, $C_4$ evaluates the false discovery rate,
$$C_4 = \frac{\vert\{(i,j): i<j, \hat s_{ij}\neq 0, \mbox{ and } s^*_{ij} = 0\}\vert}{\vert\{(i,j): i<j, s^*_{ij} = 0\}\vert}.$$
If the tuning parameter is reasonably selected, we expect that $C_2 = 1$, $C_3$ is close to 1, and $C_4$ is close to 0. 


In Figure~\ref{fig:simulation},
the averages of $C_1$ over 50 independent data sets versus the sample sizes are presented under all settings. Based on Figure~\ref{fig:simulation}, we observe that, as the sample size becomes larger,
the probability that the path of regularized estimator captures the true model increases
and is close to 1 when the sample size is over 1000.
The graphical structure is difficult to  capture   when the sample size is small.

The results of model selection based on BIC are presented in Table~\ref{table:simulation}, where the mean of $C_2$ and the
means and standard errors of $C_3$, and $C_4$ over 50 replications are presented.
According to these results, the BIC tends to choose a model that is close to the true one.
In particular, according to $C_2$,  the number of latent factors (i.e. the rank of $L^*$) can be  recovered with high probability with a reasonable sample size.
Specifically, the numbers of factors are recovered without error for all situations except when $N = 250$ for Model 3.
For this case, BIC  selects a single-factor model, which is mainly due to the small sample size.
In addition, the edges in the conditional graph are recovered with high probability according to $C_3$.
Based on $C_4$, a small number of false discoveries are observed.
In summary, the method performs well for simulated data.

\begin{figure}[t]
\centering
\includegraphics[scale = 0.5]{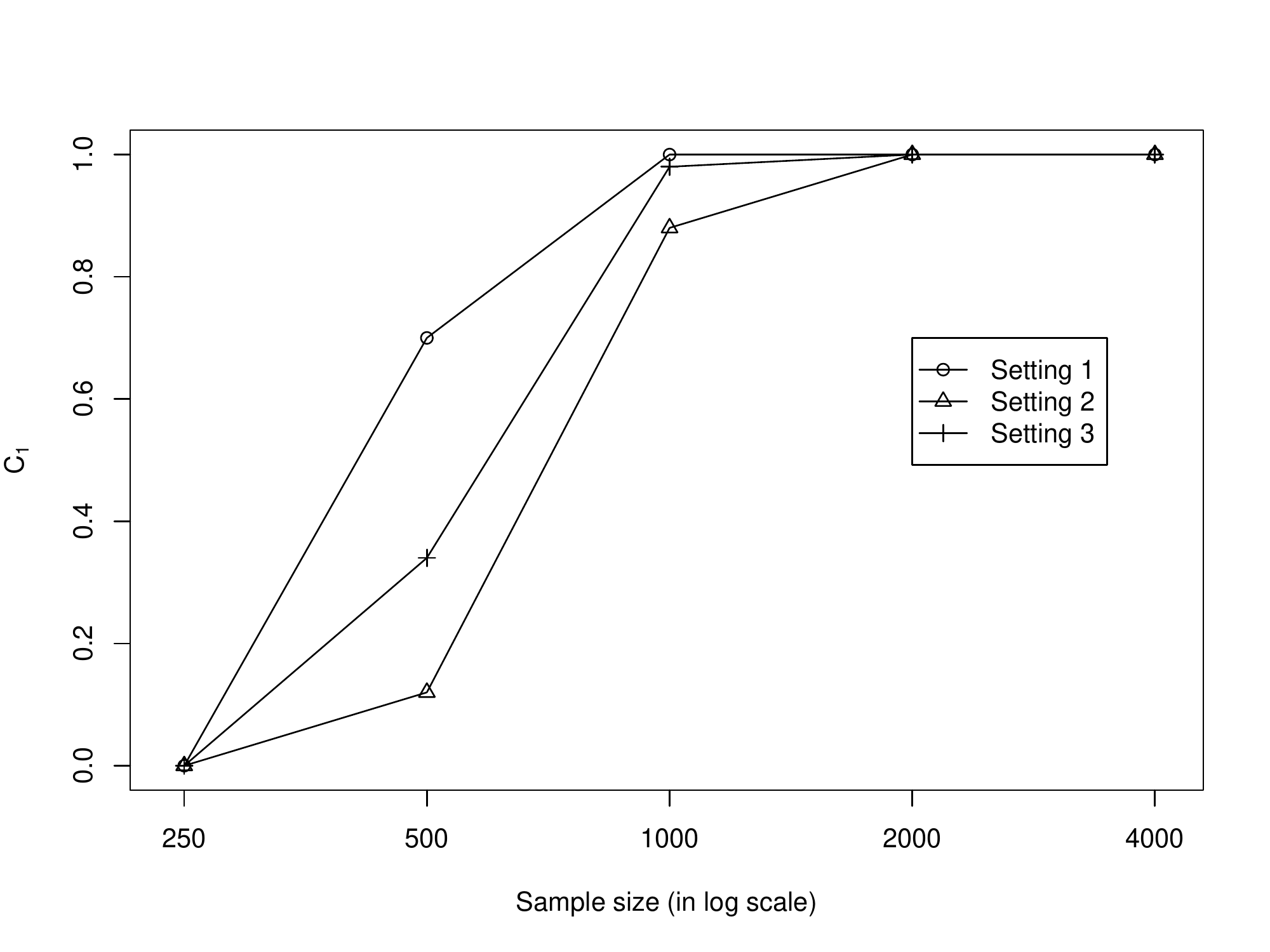}
\caption{The mean of $C_1$ over 50 replications for three settings against sample size.}
\label{fig:simulation}
\end{figure}

\begin{table}[t]
\centering
\begin{tabular}{lccccc}
\hline
$C_2$ &$N = 250$ &$N=500$ &$N=1000$& $N=2000$& $N = 4000$\\
\hline
Setting 1 &100.0 &100.0&100.0 &100.0&100.0\\
Setting 2 &100.0 &100.0 &100.0&100.0&100.0\\
Setting 3 &78.0 &100.0& 100.0 &100.0&100.0\\
\hline
\hline
$C_3$ &$N = 250$ & $N=500$ & $N=1000$& $N=2000$& $N = 4000$\\
\hline
Setting 1 &98.3(3.8) &100.0(0.0) &100.0(0.0)& 100.0(0.0)& 100.0(0.0)\\
Setting 2 &92.7(5.1) & 98.9(2.2) &100.0(0.0)&100.0(0.0)&100.0(0.0)\\
Setting 3 &94.3(6.9)&99.6(1.5)&  100(0.0)  & 100.0(0.0)& 100.0(0.0)\\
\hline
\hline
$C_4$&$N = 250$ & $N=500$ & $N=1000$& $N=2000$& $N = 4000$\\
\hline
Setting 1 &8.4(2.5)&6.7(1.6) & 5.0(1.5) &4.3(1.5) &2.8(1.1) \\
Setting 2 &8.6(2.7) & 6.2(2.4) & 0.1(0.4)& 0.0(0.1) & 0.0(0.1)\\ 
Setting 3 &6.9(2.4) &5.5(1.4) & 2.1(0.7)  &0.3(0.3)&   0.0(0.0)\\ 
\hline
\end{tabular}
\caption{The mean and standard error in percentage ($\%$) of $C_2$, $C_3$, and $C_4$.}
\label{table:simulation}
\end{table}

\subsection{Real Data Analysis}

We analyze  Eysenck's Personality Questionnaire-Revised (EPQ-R: \citealp{eysenck1985revised,eysenck2013re}).
The data set contains the responses to 79 items from
824 female respondents in the United Kingdom.
This is initially a confirmatory analysis containing three factors: Psychoticism (P), Extraversion (E), and Neuroticism (N).
Among these 79 items, 32, 23, and 24 items are designed to measure the P, E, and N factors, respectively.
The specific questions can be found in the appendix of \cite{eysenck1985revised}.
A typical item is ``Are you rather lively?". The responses are binary. The data have been preprocessed so that the negatively worded items are reversely scored (see Table~4 of \citealp{eysenck1985revised} for the scoring key).
We conduct analysis on the model goodness of fit, latent structure,  conditional graphical structure, and  their interpretations.

\paragraph{Choosing the tuning paramters.}
We optimize the tuning parameter in the range $\gamma \in (0, 0.02]$ and $\rho = \frac{\delta}{\gamma} \in (10, 20]$ on a regular lattice of size $20$ in each dimension, so that there are 400 fitted models along the solution path.
A summary of the solution path is as follows.
Among all 400 fitted models, about $11\%$ of the models have four or more factors, $57\%$ of them are three-factor models, and the rest $32\%$  have two or fewer factors.
In addition, we define the graph sparsity  level (GSP) of an estimated model as the estimated number of edges normalized by the total number of possible edges. A histogram of the GSP for all models on the path is presented in Figure~\ref{fig:GSP}. Furthermore, Figure~\ref{fig:FS} presents a box plot showing the number of factors-GSP relationship for these models, where the y axis represents the sparsity level. In Table~\ref{table:real_fit}, we list the model fitting information of the ten models that have smallest BIC values. As we can see, all ten models have three factors and have sparsity level at about $10\%$.
\begin{figure}[t]
\centering
\includegraphics[scale = 0.8]{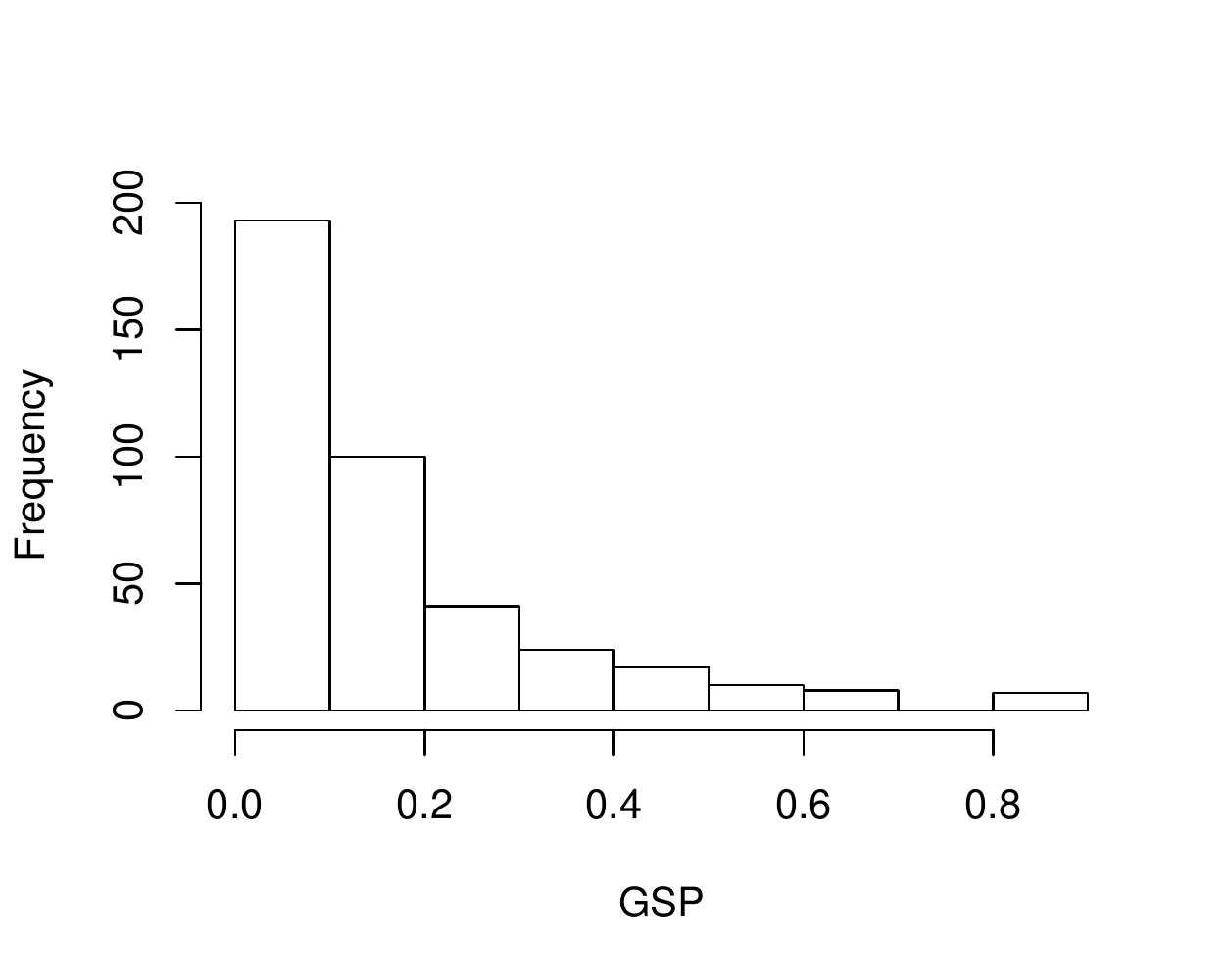}
\caption{The histogram of graph sparsity levels for all 400 models.}
\label{fig:GSP}
\end{figure}

\begin{figure}[t]
\centering
\includegraphics[scale = 0.5]{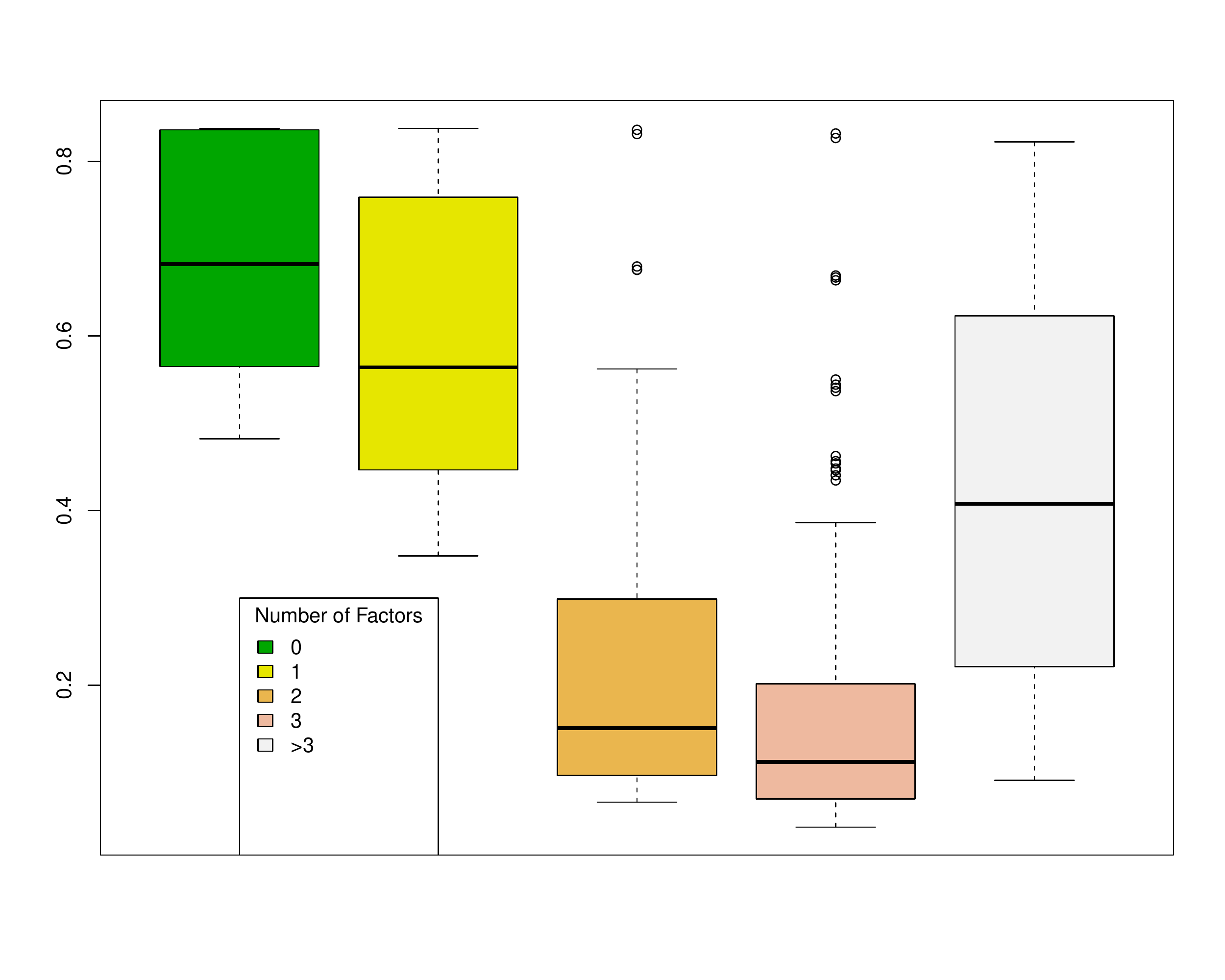}
\caption{The number of factors VS the graph sparsity level for all models.}
\label{fig:FS}
\end{figure}

\begin{table}[!ht]
\centering
\begin{tabular}{lrrrrr}
\hline
&Pseudo-lik & BIC & K& Num-edge & GSP\\
\hline
1&	-26177.6&	56779.8&	3&	346&	11\%\\
2&	-26207.9&	56800.2&	3&	340&	11\%\\
3&	-26338.4&	56806.0&	3&	302&	10\%\\
4&	-26157.1&	56812.6&	3&	357&	12\%\\
5&	-26133.0&	56831.6&	3&	367&	12\%\\
6&	-26319.5&	56835.3&	3&	312&	10\%\\
7&	-26017.4&	56855.6&	3&	405&	13\%\\
8&	-26297.9&	56859.3&	3&	322&	10\%\\
9&	-26499.1&	56872.2&	3&	264&	9\%\\
10&	-26264.1&	56872.4&	3&	334&	11\%\\
\hline
\end{tabular}
\caption{The logarithm of pseudo-likelihood, BIC, number of latent variables, number of edges, and graph sparsity levels for the top ten models.}
\label{table:real_fit}
\end{table}

\paragraph{Goodness of fit.}
Here, we investigate the model of the smallest BIC corresponding to the first model  in Table \ref{table:real_fit}.
We evaluate the goodness of fit via parametric bootstrap.
We denote $(\hat L^{sel}, \hat S^{sel})$ the maximal pseudo-likelihood estimates of the selected model.
1000 independent bootstrap data sets, each of which contains 824 samples, are generated from the latent undirected graphical model with parameters $(\hat L^{sel}, \hat S^{sel})$ via the Gibbs sampler. 
For the $b$th bootstrap data set denoted by $(\mathbf x_1^b, ..., \mathbf x_N^b)$, we compute the logarithm of unnormalized joint likelihood under the parameters $(\hat L^{sel}, \hat S^{sel})$ as
$$l_b^{sel} = \frac{1}{2} \sum _{i=1}^N (\mathbf x_i^b)^\top (\hat L^{sel} + \hat S^{sel})\mathbf x_i^b.$$
The empirical distribution of $(l_1^{sel}, ..., l_{1000}^{sel})$ is then compared with the observed one
$$l^{sel} = \frac{1}{2} \sum _{i=1}^N \mathbf x_i^\top (\hat L^{sel} + \hat S^{sel})\mathbf x_i,$$
where $(\mathbf x_1, ..., \mathbf x_N)$ is the observed responses. The normalizing constants of the joint likelihood are omitted here, because they are the same for all data sets and therefore do not play a role when comparing the observed likelihood with the corresponding bootstrap distribution.
The histogram of $(l_1^{sel}, ..., l_{1000}^{sel})$ is shown in the left panel of Figure~\ref{fig:fitting} and the observed log-likelihood $l^{sel} = -13502.4$ is marked by the red arrow with a $p$-value $=33.7\%$ suggesting that the model fits the data reasonably well.

For comparison purpose, we fit a three-dimensional IRT model in \eqref{IRT} and \eqref{local}.
This corresponds to setting the graph $S=0$ as follows
\begin{equation}\label{IRT_fit}
\begin{aligned}
(\hat L^{IRT}, \hat S^{IRT})~~ =  ~~&  \arg\max_{(L,S)}\{\mathcal L(L,S) \}\\
s.t. ~~& \text{rank}(L) = 3 \mbox{ and $L$ is positive semidefinite,} \\
& \forall i\neq j,   s_{ij} = 0.
\end{aligned}
\end{equation}
We check the goodness of fit of the three-dimensional IRT model via the same parametric bootstrap procedure, based on 1000 bootstrap samples.
The observed log-likelihood (unnormalized) is 39.4 and the corresponding bootstrap distribution is shown in the right penal of Figure~\ref{fig:fitting} with  a $p$-value $=1.7\%$ suggesting  that the three-dimensional IRT model does not fit the data well.
This comparison shows that the model fitting is substantially improved by including the additional conditional graph while maintaining a low-dimensional latent structure.

\begin{figure}[t]
\centering
\includegraphics[scale = 0.4]{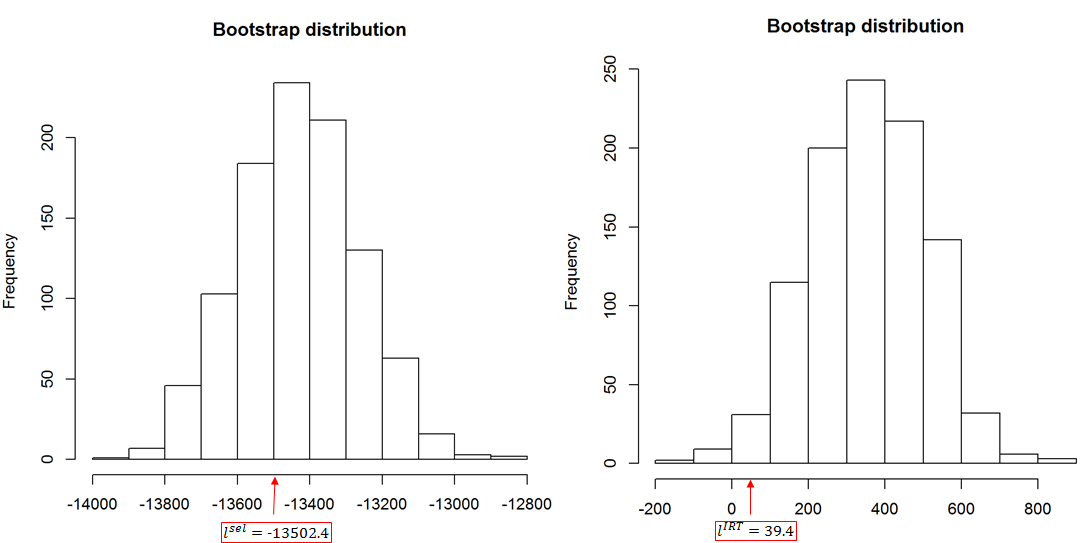}
\caption{Parametric bootstrap for checking the fit of the selected model (left) and the three-dimensional IRT model (right).}
\label{fig:fitting}
\end{figure}

\paragraph{Latent structure.}

The loading matrix $A$ is identified up to a non-degenerate rotation. Various methods are proposed to identify a particular  rotation resulting a most interpretable loading matrix \citep{thurstone1947multiple,cattell2012scientific,browne2001overview}.
Here, we adopt the varimax rotation \citep{kaiser1958varimax}, which is one of the most popular rotational methods for exploratory factor analysis.
We then check the relationship between the latent variables identified by this particular rotation and the three scales of EPQ-R. 
Based on model \eqref{JointLike}, the posterior mean of $\ttt_i$ is $E(\ttt_i \vert \XX_i = \xx_i) = \xx_i^\top A$. We replace $A$ by its estimate $\hat A$, and use $\hat \ttt_i = \hat A^\top \xx_i$ as an estimate of $\ttt_i$.
In addition, we let $T_{i}^P$, $T_i^E$, and $T_i^N$ be respondent $i$'s total scores on the P, E, and N scales respectively. In Table~\ref{tab:correlation}, the sample correlation between $(\hat \theta_{i1}, \hat \theta_{i2}, \hat \theta_{i3})$s and $(T_{i}^P, T_{i}^E, T_{i}^N)$s are calculated, where the diagonal entries being close to 1 implies that the three latent factors identified by the varimax rotation may be interpreted as Psychoticism, Extraversion, and Neuroticism, respectively.

\begin{table}[ht]
\centering
\footnotesize
\begin{tabular}{lrrr}
\hline
&P & E & N \\
\hline
$\hat \theta_1$& $\mathbf{0.92}$&  0.29 &-0.02\\
$\hat \theta_2$& 0.06&  $\mathbf{0.87}$ &-0.24\\
$\hat \theta_3$& 0.11& -0.23 & $\mathbf{0.85}$\\
\hline
\end{tabular}
\caption{The sample correlation between  $(\hat \theta_{i1}, \hat \theta_{i2}, \hat \theta_{i3})$s and $(T_{i}^P, T_{i}^E, T_{i}^N)$s. }
\label{tab:correlation}
\end{table}

%

\paragraph{Conditional graph.} For  the conditional graphical structure, the selected model has 346 edges ($\text{GSP} = 11\%$). This graph captures the association among the items that is not attributable to the latent factors.  Among the 346 edges, 91 are negative edges and 255 are positive. We investigate the positive ones. In Table~\ref{tab:item_pair}, we present the 15 item pairs that have the most positive edges.  These items share a common stimulus that is not completely attributable to the P, E, and N factors, resulting in additional dependence. For example, the first three pairs are about ``party", ``good manners",  and ``being lively", respectively. For some item pairs, the two items are essentially identical questions with different wording, such as pair 4 ``Do you stop to think things over before doing anything?" and ``Do you generally `look before you leap'?"
In addition, an item itself may be the stimulus to the other. For example, for item pair 8, it is probably that a woman would like other people to be afraid of her, because her mother is (was) not a good woman.


In addition to pairwise structures,  we also check the cliques in the estimated graph.
A clique is a subset of vertices such that every two distinct vertices in the clique are connected. A maximal clique is a clique that cannot be extended by including one more adjacent vertex. For graphical models, random variables within a clique are usually considered to be highly dependent on each other. The estimated graph has 161 maximal cliques that have at least three vertexes, including one 5-vertex clique, 32 4-vertex cliques, and 128 3-vertex cliques.
In Table~\ref{tab:item_clique}, we present the 5-vertex clique,  two 4-vertex cliques, and two 3-vertex cliques. These 4-vertex and 3-vertex cliques are the ones with the highest within-clique sum of $\hat s_{ij}$.
We also observe that the maximal cliques  identify meaningful item clusters. For example, the five cliques in Table~\ref{tab:item_clique} are about ``communication with others", ``thinking before action", ``being nervous", ``good manners", and ``meeting people", respectively.

In summary, the proposed FLaG model fits the EPQ-R data well, while a three-factor multidimensional IRT model (with local independence) has substantial lack of fit. Thus, the FLaG model improves model fitting and maintains a low-dimensional latent structure.
In addition, the estimated latent structure is very interpretable and coincides the three factors proposed in the initial confirmatory factor analysis.
Lastly, we also investigate the conditional graph that yields meaningful clusters of item in addition to the dependence induced by the latent factors.

%

\begin{table}[ht]
\centering
\footnotesize
\begin{tabular}{llllll}
\hline
& $\hat s_{ij}$& Item & Scale  & Item content\\
\hline
1&3.31&51&E&Can you easily get some life into a rather dull party?\\
&&78&E&Can you get a patty going? \\
2&2.43&21&P&(R)Are good manners very important?\\
&&41&P&(R)Do good manners and cleanliness matter much to you? \\
3&2.32&11&E&Are you rather lively? \\
&&94&E&Do other people think of you as being very lively? \\
4&2.19&2&P&(R)Do you stop to think things over before doing anything? \\
&&81&P&(R)Do you generally `look before you leap'?\\
5&2.73&22&N&Are your feelings easily hurt?\\
&&87&N&Are you easily hurt when people find fault with you ot the work you do?\\
6&1.97&35&N&Would you call yourself a nervous person? \\
&&83&N&Do you suffer from `nerves'?\\
7&1.83&6&E&Are you a talkative person?\\
&&47&E&(R)Are you mostly quiet when you are with other people?\\
8&1.81&91&P&Would you like other people to be afraid of you?\\
&&68&P&(R)Is (or was) your mother a good woman? \\
9&1.70&34&P&Do you have enemies who want to harm you?\\
&&73&P&Are there several people who keep trying to avoid you? \\
10&1.69&24&E&(R)Do you tend to keep in the background on social occasions? \\
&&47&E&(R)Are you mostly quiet when you are with other people?\\
11&1.67&13&N&Do you often worry about things you should not have done or said? \\
&&31&N&Are you often troubled about feelings of guilt?\\
12&1.67&13&N&Do you often worry about things you should not have done or said? \\
&&80&N&Do you worry too long after an embarrassing experience?\\
13	&	1.67	&	95	&	P	&	Do people tell you a lot of lies?	\\
	&		&	85	&	P	&	Can you on the whole trust people to tell the truth? 	\\
14	&	1.61	&	20&	E	&	Do you enjoy meeting new people?	\\
	&		&	58&	E	&	Do you like mixing with people?	\\
15	&	1.50	&	58	&	E	&	Do you like mixing with people?	\\
	&		&	33	&	E	&	Do you prefer reading to meeting people? 	\\
\hline
\end{tabular}
\caption{The top 15 item pairs corresponding to the most positive edges. The item ID is consistent with \cite{eysenck1985revised} and the reversely scored items are marked by ``(R)". }
\label{tab:item_pair}
\end{table}

%

\begin{table}[ht]
\centering
\footnotesize
\begin{tabular}{llllll}
\hline
& Item & Scale  & Item content\\
\hline
1&6&E&Are you a talkative person?\\
&94&E&Do other people think of you as being very lively? \\
&47&E&(R)Are you mostly quiet when you are with other people?\\
&24&E&(R)Do you tend to keep in the background on social occasions? \\
&63&E&Do you nearly always have a `ready answer' when people talk to you? \\
\hline
2&81&P&(R)Do you generally `look before you leap'?\\
&2&P&(R)Do you stop to think things over before doing anything? \\
&69&E&Do you often make decisions on the spur of the moment?\\
&61&E&Have people said that you sometimes act too rashly?\\
\hline
3&35&N&Would you call yourself a nervous person? \\
&38&N&Are you a worrier?\\
&46&N&Would you call yourself tense or ‘highly-strung’?\\
&83&N&Do you suffer from ‘nerves’?\\
\hline
4&21&P&(R)Are good manners very important?\\
&14&P&(R)Do you dislike people who don’t know how to behave themselves?\\
&41&P&(R)Do good manners and cleanliness matter much to you? \\
\hline
5&20&E&Do you enjoy meeting new people?\\
&33&E&(R)Do you prefer reading to meeting people? \\
&58&E&Do you like mixing with people?\\
\hline
\end{tabular}
\caption{Examples of maximal cliques of the estimated graph. The item ID is consistent with \cite{eysenck1985revised} and the reversely scored items are marked by ``(R)".  }
\label{tab:item_clique}
\end{table}

\section{Conclusion and Discussion}
The main contribution of this paper is three-fold.
First, we propose a fused latent and graphical (FLaG) model by combining  a multidimensional item response model and the Ising model. Then, we consider the regularized pseudo-likelihood by means of the $L_1$ and nuclear norm penalties. Lastly, the computation of the regularized estimator is facilitated by our developing an algorithm based on the alternating direction method of multiplier
 to optimize a non-smooth and convex objective function.

We propose to use the BIC for the tuning parameter selection, which  performs well empirically.
	The proposed method is applied to a real data set based on the revised Eysenck's Personality Questionnaire that consists of items designed to measure Psychoticism, Extraversion, and Neuroticism. The estimated model receives good interpretation.
In particular, the estimated three latent variables correspond to the well known Psychoticism, Extraversion, and Neuroticism personality factors, respectively.
In addition, there are a significant number of edges in the conditional graphical model, which indicates the inadequacy of a traditional three-factor IRT model assuming local independence. 		
This is also confirmed by a quantitative model diagnosis via the parametric bootstrap.
Finally, the conditional graph provides us a better understanding of the items, which may be utilized to improve the questionnaire design.

	



\bibliographystyle{apa}
\bibliography{SL,CDM}

\newpage

\appendix
\section{Proof of Theorem~\ref{thm:consistent}}
	Throughout the proof, we will use $\kappa$ as generic notation for large and not-so-important constants whose value may vary from place to place. Similarly, we use $\varepsilon$ as generic notation for small positive constants.
	Furthermore, for two sequences of random variables $a_N$ and $b_N$, we write $a_N=o_P(b_N)$ if $b_N/a_N{\to} 0$ in probability and $a_N=O_P(b_N)$ if $a_N/b_N$ is tight. We also use the notation ``$=_P$'', ``$<_P$", ``$\leq_P$", ``$>_P$" and ``$\geq_P$" to indicate the equality and inequalities hold with a probability converging to one as $N$ goes to infinity.
\paragraph{Proof Strategy.}
To assist the readers, we first provide a  sketch of the proof for the theorem. We introduce several notation and definitions.
Let the  eigendecomposition of $L^*$ be $L^*= U^* D^* {U^*}^{\top}$, such that $U^*$ is a $J\times J$ orthogonal matrix and $D^*$  is a $J\times J$ diagonal matrix whose first $K$ diagonal elements are strictly positive.
 We write $U^*=[U_1^*,U_2^*]$ where $U_1^*$ is the first $K$ columns of $U^*$. Let $D_1^*$ be the $K\times K$ diagonal matrix  containing the nonzero diagonal elements of $D^*$. Define the localization set
\begin{eqnarray*}
\mathcal{M}_1 = \{(S,L):
& S=S^*+\delta_S,& \|\delta_S\|_{\infty}\leq \cs, \delta_S \mbox{ is symmetric,} \\
 ~~~~~~~~&L= UDU^{\top},& \|U-U^*\|_{\infty}\leq \cu, ~U \mbox{ is a $J\times J$ orthogonal matrix},\\
 &&\|D-D^*\|_{\infty}\leq \cd, \mbox{ and }D \mbox{ is a $J\times J$ diagonal matrix}\},
 \end{eqnarray*}
and a subset
\begin{equation*}
\begin{aligned}
\mathcal{M}_2 = \{(S,L):
&~~S=S^*+{\delta_{S}}, ~~ \|\delta_S\|_{\infty}\leq \cs, \delta_S\in \MS
\\
 &~~L= U_1D_1U_1^{\top}, ~ \|U_1-U_1^*\|_{\infty}\leq \cu, U_1 \mbox{ is a $J\times K$  matrix, } U_1^{\top}U_1=I_K,\\
 &  ~~~~~~~~~~~~~~~~~~~~~~ \|D_1-D_1^*\|_{\infty}\leq \cd, \mbox{ and }D_1 \mbox{ is a $K\times K$ diagonal matrix}
 \}.
 \end{aligned}
 \end{equation*}
Here, $\eta$ is a positive constant that is sufficiently small.
For each pair of $(S,L)\in \mathcal{M}_1$, $S$ is close to $S^*$. Moreover, the eigendecomposition of $L$ and $L^*$  are close to each other. As the sample size $N$ grows large, the set $\mathcal{M}_1$ will tend to $\{(S^*,L^*)\}$. The set $\mathcal{M}_2$ is a subset of $\mathcal{M}_1$.
For each pair of $(S,L)\in\mathcal{M}_2$, $S$ has the same sparsity pattern as $S^*$, and $L$ has the same rank as $L^*$ for sufficiently large $N$.

The proof consists of two steps.
\begin{enumerate}
\item We first prove that with a probability converging to $1$ the optimization problem \eqref{est} restricted to the subset $\mathcal{M}_2$ has a unique solution, which does not lie on the manifold boundary of $\mathcal{M}_2$. This part of proof is presented in Section~\ref{sec:proof-step1}.
\item We then show the unique solution restricted to $\mathcal{M}_2$ is also a solution to  \eqref{est} on $\mathcal{M}_1$. It is further shown that with probability converging to $1$ this solution is the unique solution to \eqref{est} restricted to $\mathcal{M}_1$.
 This part of proof is presented in Section~\ref{sec:proof-step2}.
\end{enumerate}
The previous two steps together imply that the convex optimization problem \eqref{est} with the constraint $(S,L)\in\mathcal{M}_1$ has a unique solution  that belongs to  $\mathcal{M}_2$. Furthermore, this solution is an interior point of  $\mathcal{M}_1$. Thanks to the convexity of the objective function, $(\hat S,\hat L)$ is also the unique solution to the optimization problem \eqref{est}. We conclude the proof by noticing that all $(S,L)\in \mathcal{M}_2$ converge to the true parameter $(S^*,L^*)$ as $N\to\infty$ with the same sparsity and low rank structure.
\subsection{Proof step 1}\label{sec:proof-step1}
Denote by $(\hat{S}_{\mathcal{M}_2},\hat{L}_{\mathcal{M}_2})$ a solution to the optimization problem
\begin{equation}
\begin{aligned}
&\min \big\{h_N(S+L)+ \gamma_N\|\OO(S)\|_1+\delta_N\|L\|_*\big\} \label{eq:optim-m2}\\
&\mbox{ subject to $L$ is positive semidefinite and $S$ is symmetric and }{(L,S)\in \mathcal{M}_2}.
\end{aligned}
\end{equation}
Recall here that the function $h_N$ is defined as in \eqref{eq:def:h_fun}.
We write the eigendecomposition $\hat{L}_{\MM_2}=\UM\DM\UM^{\top}$, where $\UM^{\top}\UM=I_K$ and $\DM$ is a $K\times K$ diagonal matrix.
To establish that $(\SM,\LM)$ does not lie on the manifold boundary of $\mathcal{M}_2$, it is sufficient to show that
\begin{eqnarray}
\|\hat{S}_{\mathcal{M}_2}-S^*\|_{\infty}&<_P&\cs\label{eq:s-bound}\\
\|\DM-D_1^*\|_{\infty}&<_P&\cd.\label{eq:d-bound}\\
\|\UM-U^*_1\|_{\infty}&<_P&\cu.\label{eq:u-bound}
\end{eqnarray}
To start with, we present a useful lemma.
\begin{lemma}\label{lemma:d}
	Let $$\hat{\MD}=\{\UM D'_1 \UM^\top: D'_1 \mbox{ is a $K\times K$ diagonal matrix}\}.$$
	Consider the convex optimization problem
	\begin{equation}
	\min_{S\in \MS, L\in \hat{\MD}} \big\{h_N(S+L)+ \gamma_N\|\OO(S)\|_1+\delta_N\|L\|_*\big\}. \label{eq:optim-d}
	\end{equation}
	Then  \eqref{eq:optim-d} has a unique solution with probability converging to $1$.
	Denote the solution by $(\hat{S}_{\hat{\MD}},\LD)$ and $
	\hat{L}_{\hat{\MD}}= \UM \DDD \UM^{\top}.
	$
	In addition, there exists a constant $\kappa>0$ such that
	$
	\|\hat{S}_{\hat{\MD}}-S^*\|_{\infty}\leq_P \kappa\cu
	$
	and
	$
	\|\DDD -D_1^*\|_{\infty}\leq_P \kappa \cu.
	$
\end{lemma}
Because of the convexity of the objective function, a direct application of the above lemma is
$$(\SM,\LM)=_P(\SD,\LD) \mbox{ and }\DM=_P\DDD.$$
Thus, \eqref{eq:s-bound} and \eqref{eq:d-bound} are proved.
We show  \eqref{eq:u-bound} by contradiction. If on the contrary $\|\UM-U_1^*\|_{\infty}=\cu$, then in what follows we will show that
\begin{equation}\label{eq:contradict}
h_N(\SD+\LD)+\gamma_N\|\OO(\SD)\|_1+\delta_N\|\LD\|_{*}>_P h_N(S^*+L^*)+\gamma_N\|\OO(S^*)\|_1+\delta_N\|L^*\|_{*}
\end{equation}
and thus a contradiction is reached.
We start with
the Taylor expansion of $h_N(S+L)$ around $S^*$ and $L^*$. Let $h(M)=\mathbb{E} h_N(M)$, then we have
\begin{equation}\label{eq:taylor}
h_N(M^*+\Delta)=h(M^*)+ \frac{1}{2}v(\Delta)^{\top}\II v(\Delta)+R_N(\Delta),
\end{equation}
where $M^*=S^*+L^*$, and the function
$v:\mathbb{R}^J\times \mathbb{R}^J\to \mathbb{R}^{J^2}\times 1$ is a map that vectorizes a matrix. Moreover,  $R_N(\Delta)$ is the remainder term satisfying
\begin{equation}\label{eq:remain-bound}
R_N(\Delta)=R_N(\mathbf{0}_{J\times J})+O_P(\|\Delta\|_{\infty}^3)+O_P(\frac{\|\Delta\|_{\infty}}{\sqrt{N}})\mbox{ as }\Delta\to 0, N\to \infty,
\end{equation}
where $R_N(\mathbf{0}_{J\times J})=h_N(M^*) - h(M^*)$,  $O_P({\|\Delta\|_{\infty}}/{\sqrt{N}})$ term corresponds to $v(\nabla h_N(M^*))^{\top} v(\Delta)$, and $O_P(\|\Delta\|_{\infty}^3)$ characterizes the remainder.
Furthermore, as $\Delta\to 0$ the second derivative satisfies
\begin{equation}\label{eq:nabla-2-remain-bound}
\nabla^2{R}_N(\Delta) = O_P(\frac{1}{\sqrt{N}}+\|\Delta\|_{\infty}),
\end{equation}
where the $O_{P}(\|\Delta \|_{\infty})$ term corresponds to $\nabla^2 h(M^*+\Delta)-\II$, and $O_P({1}/{\sqrt{N}})$ corresponds to $\nabla^2 h_N(M^*+\Delta) - \nabla^2 h(M^*+\Delta)$. This further implies that the first derivative of $R_N(\Delta)$ satisfies
\begin{equation}\label{eq:nabla-remain-bound}
\nabla R_N(\Delta) = O(\|\Delta\|_{\infty}^2)+O_P(\frac{1}{\sqrt{N}}) \mbox{ as }\Delta\to 0.
\end{equation}
We plug $\Delta=\SD+\LD-S^*-L^*$ into \eqref{eq:taylor}, then
\begin{multline}\label{eq:talor-d}
h_N(\SD+\LD)-h_N(S^*+L^*)\\= \frac{1}{2}v(\SD+\LD-(S^*+L^*))^\top\II v(\SD+\LD-(S^*+L^*) )+R_N(\SD+\LD-(S^*+L^*))-R_N(\mathbf{0}_{J\times J}).
\end{multline}
We first establish a lower bound for $R_N(\SD+\LD-(S^*+L^*))-R_N(\mathbf{0}_{J\times J})$.
According to Lemma~\ref{lemma:d},
we have
\begin{equation}\label{eq:deltaupper}
\|\SD+\LD-(S^*+L^*)\|_{\infty}\leq_P \kappa \cu,
\end{equation}
with a possibly different $\kappa$.
The above display and \eqref{eq:remain-bound} yield
\begin{equation}\label{eq:upper-remain}
R_N(\SD+\LD-(S^*+L^*))-R_N(\mathbf{0}_{J\times J})= O(\cuqb)+O_P(\frac{\cu}{\sqrt{N}}).
\end{equation}
	We proceed to a lower bound for the term $\frac{1}{2}v(\SD+\LD-(S^*+L^*))^\top\II v(\SD+\LD-(S^*+L^*) )$ on the right-hand side of \eqref{eq:talor-d} with the aid of the following two lemmas.
\begin{lemma}\label{lemma:trans}
Under Assumption A3,
 there exists a positive constant $\varepsilon$ such that
\begin{equation*}\label{eq:trans}
\|S+L\|_{\infty}\geq \varepsilon \|L\|_{\infty} \mbox{ for all } (S, L)\in \MS\times T_{L^*}\MFL.
\end{equation*}
\end{lemma}
\begin{lemma}\label{lemma:L}
Let
\begin{equation}\label{eq:deltaL}
\Delta_L=U_1^*D_1^*(\UM-U_1^*)^{\top} +  (\UM-U_1^*)D_1^* U_1^{*\top} + U_1^* (\DDD-D_1^*) U_1^{*\top}.
\end{equation} Then, we have
\begin{itemize}
	\item [(i)] $\Delta_L\in T_{L^*}\MFL$.
	\item [(ii)] There exists positive constant $\varepsilon$ such that $\|\Delta_L\|_{\infty}>_P \varepsilon \cu$, for all  $\|\UM-U_1^*\|_{\infty}=\cu$.
	\item [(iii)] $
	\|\LD -L^*-\Delta_L\|_{\infty}\leq_P \kappa\cusq
	$.
\end{itemize}
\end{lemma}
According to Lemma~\ref{lemma:trans} and Lemma \ref{lemma:L}(i)(iii) and  noticing that $\SD-S^*\in \MS$, we have
$$
\|\SD+\LD-(S^*+L^*)\|_{\infty}\geq_P \|\SD-S^*+\Delta_L\|_{\infty} - \kappa \cusq \geq_P \varepsilon \|\Delta_L \|_{\infty} - \kappa\cusq .
$$
According to Lemma~\ref{lemma:L}(ii), the above display further implies that
\begin{equation}\label{eq:lowerdelta}
\|\SD+\LD-(S^*+L^*)\|_{\infty}>_P\varepsilon \cu,
\end{equation}
with a possibly different $\varepsilon$.
According to assumption A1, $\II$ is positive definite. Therefore,
\begin{equation}\label{eq:lower-quadratic}
v(\SD+\LD-(S^*+L^*))^{\top}\II v(\SD+\LD-(S^*+L^*)  )>\varepsilon \|\SD+\LD-(S^*+L^*) \|\geq_P \varepsilon^2 \cusq .
\end{equation}
The second inequality of the above display is due to \eqref{eq:lowerdelta}.
 \eqref{eq:talor-d}, \eqref{eq:upper-remain} and \eqref{eq:lower-quadratic} give
\begin{equation}\label{eq:h-difference}
h_N(\SD+\LD)-h_N(S^*+L^*)>_P \frac{\varepsilon^2}{2} \cusq .
\end{equation}
We proceed to the regularization terms in \eqref{eq:contradict}. For the $L_1$ penalty term, we have
\begin{equation}\label{eq:s-l1-bound}
\|\OO(\SD)\|_{1}-\|\OO(S^*)\|_1= \mathrm{sign}(\OO(S^*))\cdot (\SD-S^*)=_PO(\cu).
\end{equation}
The second equality in the above display is due to Lemma~\ref{lemma:d}.
For the nuclear norm term, we have
	$$
	\|\LD\|_{*}-\|L^*\|_*\geq -\|\LD-L^* \|_{*}\geq -\kappa\cu .
	$$
Again, the second inequality in the above display is due to Lemma~\ref{lemma:d}.
	Notice that $\delta_N=\rho\gamma_N$.
Equations \eqref{eq:h-difference},  \eqref{eq:s-l1-bound} and the above inequality imply
\begin{align*}
&h_N(\SD+\LD)-h_N(S^*+L^*)+\gamma_N (\|\OO(\SD)\|_1-\|\OO(S^*)\|_1)+\delta_N (\|\LD\|_{*}-\|L^*\|_*)\\
>_P
&~\varepsilon\cusq >_P0,
\end{align*}
with a possibly different $\varepsilon$.
Notice that $(\SD,\LD)=(\SM,\LM)$, so we obtain \eqref{eq:contradict} by rearranging terms in the above inequality, and this contradicts the definition of $\SM$ and $\LM$. This completes the proof for \eqref{eq:u-bound}.
Thus, $(\SM,\LM)$ is an interior point of $\mathcal{M}_2$.
The uniqueness of the solution is obtained according to the following lemma.
\begin{lemma}\label{lemma:solution-approx}
	The solution to the optimization problem \eqref{eq:optim-m2} is unique with a probability converging to $1$. In addition,
	\begin{equation}\label{eq:solution-approx}
	(\SM,\LM)=(S^*,L^*)+{\FF}^{-1}(\gamma_N \mathrm{sign}(\OO(S^*)),\delta_N U_1^* U_1^{*\top})+o_P(\gamma_N),
	\end{equation}
	as $N\to\infty$.
\end{lemma}
%
%
%
%
%

\subsection{Proof step 2}\label{sec:proof-step2}
In this section, we first show that $(\SM,\LM)$ is a solution of the optimization problem
\begin{equation}
\begin{aligned}
	&\min h_N(S+L)+ \gamma_N\|\OO(S)\|_1+\delta_N\|L\|_*, \label{eq:optim-m1}\\
	&\mbox{ subject to $L$ is positive semidefinite and $S$ is symmetric and } {(L,S)\in \MM_1}.
\end{aligned}
\end{equation}
To prove this, we will show that $(\SM,\LM)$ satisfies  the first order condition
\begin{equation}\label{eq:first-order-condition}
\mathbf{0}_{J\times J}\in \partial_S H|_{(\SM,\LM)}\mbox{ and } \mathbf{0}_{J\times J} \in \partial_L H|_{(\SM,\LM)},
\end{equation}
where the function $H$ is the objective function
\begin{equation}\label{eq:objective}
H(S,L)=h_N(S+L)+\gamma_N\|\OO(S)\|_1+\delta_N\|L\|_*
\end{equation}
and $\partial_S H$ and $\partial_L H$ denotes the sub-differentials of $H$. See \cite{rockafellar2015convex} for more details of sub-differentials of convex functions. We first derive an explicit expression of the first order condition.
The sub-differential with respect to $S$ is defined as
\begin{equation}\label{eq:subdifferntial-s}
\partial_S H|_{(\SM,\LM)} =\{ \nabla{h}_N(\SM+\LM)+\gamma_N (\mathrm{sign}(\OO(S^*)))+\gamma_N W: \|W\|_{\infty}\leq 1 \mbox{ and } W\in \MS^{\bot}\},
\end{equation}
where $\MS^{\bot}$ is the orthogonal complement space of $\MS$ in the space of symmetric matrices.
 According to Example 2 of \cite{watson1992characterization}, the sub-differential with respect to $L$ is
\begin{equation}\label{eq:subdifferential-l}
\partial_{L}H|_{(\SM,\LM)}=\{\nabla{h}_N(\SM+\LM)+\delta_N\UM\UM^{\top}+\delta_N\UCM W \UCM^{\top}: \|W\|_2\leq 1\},
\end{equation}
where $\UCM$ is a $J\times (J-K)$ matrix satisfying $\UM^{\top} \UCM=\mathbf{0}_{K\times (J-K)}$, $ \UCM^{\top}\UCM=I_{J-K} $ and $\|\UCM-U_2^*\|_\infty\leq \cu$.
For some $(S,L)$, if
$$
\PP_{\MS} (S)=\mathbf{0}_{J\times J},~ \PP_{\MS^{\bot}} (S)=\mathbf{0}_{J\times J}, ~\PP_{T_{{\LM}}\MFL}(L)=\mathbf{0}_{J\times J} \mbox{ and }\PP_{(T_{{\LM}}\MFL)^{\bot}}(L)=\mathbf{0}_{J\times J},
$$
then $S=\mathbf{0}_{J\times J}$ and $L=\mathbf{0}_{J\times J}$.
Consequently, to prove \eqref{eq:first-order-condition}, it suffices to show that
\begin{equation}\label{eq:first-condition-smooth}
\PP_{\MS} \partial_S H|_{(\SM,\LM)} =\{\mathbf{0}_{J\times J}\}\mbox{ and }\PP_{T_{{\LM}}\MFL} \partial_L H|_{(\SM,\LM)}=\{\mathbf{0}_{J\times J}\},
\end{equation}
and
\begin{equation}\label{eq:first-condition-orth}
\mathbf{0}_{J\times J}\in \PP_{\MS^{\bot}} \partial_S H|_{(\SM,\LM)} \mbox{ and }\mathbf{0}_{J\times J}\in \PP_{(T_{{\LM}}\MFL)^{\bot}} \partial_L H|_{(\SM,\LM)}.
\end{equation}
According to the definition of $(\SM,\LM)$, it
is the solution to the optimization \eqref{eq:optim-m2}. In addition, according to the discussion in Section~\ref{sec:proof-step1}, $(\SM,\LM)$ does not lie on the boundary of $\MM_2$. Therefore, it
satisfies the first order condition of \eqref{eq:optim-m2}, which is equivalent to \eqref{eq:first-condition-smooth}.
Thus, to prove \eqref{eq:first-order-condition} it is sufficient to show \eqref{eq:first-condition-orth}.
The next lemma establishes an equivalent expression for \eqref{eq:first-condition-orth}.
\begin{lemma}\label{lemma:equivalent}
\eqref{eq:first-condition-orth} is equivalent to
\begin{equation}\label{eq:sufficient-step-2}
\|\PP_{\MS^{\bot}}\nabla h_N(\SM+\LM)\|_{\infty} \leq \gamma_N \mbox{ and } \|\PP_{(T_{{\LM}}\MFL)^{\bot}}\nabla h_N(\SM+\LM) \|_{2}\leq \delta_N.
\end{equation}
\end{lemma}
We proceed to proving \eqref{eq:sufficient-step-2}.
Take gradient on both side of \eqref{eq:taylor} to obtain
\begin{equation}\label{eq:taylor-nabla}
\nabla h_N(S^*+L^*+\Delta) = \II v(\Delta) + \nabla{R}_N(\Delta).
\end{equation}
We plug $\Delta= \SM+\LM-S^*-L^*$ into the above equation to get
\begin{equation}\label{eq:taylor-nabla-h}
\nabla h_N(\SM+\LM) = \II v(\SM+\LM-S^*-L^*)+\nabla{R}_N(\SM+\LM-S^*-L^*).
\end{equation}
According to Lemma~\ref{lemma:solution-approx},
$$
\SM+\LM-S^*-L^*= \bA\FF^{-1}(\gamma_N \mathrm{sign}(\OO(S^*)),\delta_N U_1^* U_1^{*\top})+o_{P}(\gamma_N),
$$
where $\bA$ is the adding operator of two matrices $\bA(A,B)=A+B$ and $\FF$ is the operator defined as in \eqref{F}.
Combining this with \eqref{eq:nabla-remain-bound}, \eqref{eq:taylor-nabla-h}, and notice that $\delta_N=\rho\gamma_N$, we have
\begin{equation*}
\nabla h_N(\SM+\LM)= \gamma_N\II\bA\FF^{-1}( \mathrm{sign}(\OO(S^*)),\rho U_1^* U_1^{*\top})+o_{P}(\gamma_N).
\end{equation*}
and consequently,
\begin{eqnarray}
\PP_{\MS^{\bot}}\nabla h_N(\SM+\LM)&=& \gamma_N\PP_{\MS^{\bot}}\II\bA\FF^{-1}( \mathrm{sign}(\OO(S^*)),\rho U_1^* U_1^{*\top})+o_{P}(\gamma_N),\notag\\
\PP_{{T_{\LM}\MFL}^{\bot}}\nabla h_N(\SM+\LM)&=&\gamma_N\PP_{\mathcal{T_{L^*}\MFL}^{\bot}}\II\bA\FF^{-1}( \mathrm{sign}(\OO(S^*)),\rho U_1^* U_1^{*\top})+o_{P}(\gamma_N).\notag\\
&&\label{eq:gradient-hat}
\end{eqnarray}
We complete the proof for $(\SM,\LM)$ to be a solution of \eqref{eq:optim-m1}
by noticing that \eqref{eq:sufficient-step-2} is a direct application of Assumption A4  and the above equation.
We proceed to the proof of the uniqueness of the solution to \eqref{eq:optim-m1}. Because the objective function $H(S,L)$ is a convex function, it is sufficient to show the uniqueness of the solution in a neighborhood of $(\SM,\LM)$. We choose a small neighborhood as follows:
\begin{equation*}\begin{aligned}
\mathcal{N}=
\Big\{(S,L):&
~~\|{S} - \SM\|_{\infty}< \ce, \mbox{ $S$ is symmetric,}\\
& ~~\mbox{$L$ has the eigendecomposition }
{L} = \begin{bmatrix}
U_1, U_2
\end{bmatrix}
\begin{bmatrix}
D_1 & \mathbf{0}_{J\times (J-K)}\\
\mathbf{0}_{(J-K)\times J} & D_2
\end{bmatrix}
\begin{bmatrix}
U_1,U_2
\end{bmatrix}^{\top},\\
&
 ~~\|U_1-\UM\|_{\infty}<\ce, \|U_2-\UMtwo  \|_{\infty}<\ce,\\
 &~~\|D_1-\DM \|_{\infty}<\ce, \|D_2\|_{\infty}<\ce ~~~~~~~~~~~~~~~~~~~~~~~~~~~~~~~~~~~~~~~~~~~~~~~~~~
 \Big\}. \label{eq:tilde-s-range}
\end{aligned}
\end{equation*}
The next lemma, together with the uniqueness of solution to \eqref{eq:optim-m2} established in Lemma~\ref{lemma:solution-approx}, guarantees that $(\SM,\LM)$ is the unique solution in $\mathcal{N}$.
\begin{lemma}\label{lemma:another}
For all $(\SA,\LA)\in\mathcal{N}$, if $(\SA,\LA)$ is a solution to \eqref{eq:optim-m1}, then $(\SA,\LA)\in\MM_2$.
\end{lemma}

\section{Proof of the supporting lemmas}\label{sec:proof-lemma}
\begin{proof}[Proof of Lemma~\ref{lemma:f-inv}]
We prove the lemma by contradiction. If on the contrary, $\FF$ is not invertible over $\MS\times T_{L^*}\mathfrak L$, then there exists $(S,L)\in \MS\times T_{L^*}\mathfrak L$ such that
\begin{equation}\label{eq：f-injective}
(S,L)\neq (\mathbf{0}_{J\times J},\mathbf{0}_{J\times J}) \mbox{ and }
\FF(S,L)=(\mathbf{0}_{J\times J},\mathbf{0}_{J\times J}).
\end{equation}
Recall that $\FF$ is defined as
$$
	\FF(S,L) = (\PP_{\MS}\{ \II(S+L)\}, \PP_{T_{L^*}\mathfrak L}\{ \II(S+L)\}).
$$
Then, \eqref{eq：f-injective} implies that
\begin{equation}\label{eq:isl}
\II(S+L) \in \MS^{\bot}\cap ({T_{L^*}\mathfrak L})^{\bot}.
\end{equation}
Consequently,
$$
v(S+L)^{\top}\II v(S+L)= (S+L)\cdot \II(S+L) =  S \cdot W +L \cdot W=0,
$$
where we define $W=\II(S+L) $ and the last equality is due to \eqref{eq:isl}. According Assumption A1, $\II$ is positive semidefinite. Thus, the above display implies that $S+L=\mathbf{0}_{J\times J}$. According to Assumption A3, this further implies that $S=\mathbf{0}_{J\times J}$ and $L=\mathbf{0}_{J\times J}$. Note that this contradicts  our assumption that $(S,L)\neq (\mathbf{0}_{J\times J},\mathbf{0}_{J\times J})$.
\end{proof}

\bigskip

\begin{proof}[Proof of Lemma~\ref{lemma:d}]
We consider the first order condition for the optimization problem \eqref{eq:optim-d}. Notice that $\MS$ and $\hat{\MD}$ are linear spaces, so the first order condition becomes
$$
\mathbf{0}_{J\times J} \in \PP_{\MS}\partial_{S}H|_{(S,L)} \mbox{ and } \mathbf{0}_{J\times J} \in \PP_{\hat{\MD}}\partial_{L}H|_{(S,L)},
$$
where $H$ is defined in \eqref{eq:objective}. We will show that there is a unique $(S,L)\in \MS\times\hat{\MD}$ satisfying the first order condition. Because of  the convexity of the optimization problem \eqref{eq:optim-d}, it suffices to show that with a probability converging to $1$ there is a unique $(S,L)\in\mathcal{B}$ satisfying the first order condition, where
$$
\mathcal{B} = \{(S,L)\in\MS\times\hat{\MD}:\|S-S^*\|_{\infty}\leq \cu, \mbox{ and } \|L-L^*\|_{\infty}\leq \cu\}.
$$ We simplify the first order condition for  $(S,L)\in \mathcal{B}$.
For the $L_1$ penalty term, if $\|S-S^*\|_{\infty}\leq \cu$ and $S\in\MS$, then $\|\OO(S)\|_1$ is smooth on  $\mathcal{S}^*$ and
\begin{equation}\label{eq:l1}
\PP_{\MS}\partial_{S}\|\OO(S)\|_1 = \mathrm{sign}(\OO(S^*)) \mbox{ for } S\in \MS.
\end{equation}
Similarly, for $L\in \hat{\MD}$ and $\|L-L^*\|_\infty\leq \cu$, $\|L\|_*$ is smooth over the linear space $\hat{\MD}$ and
\begin{equation}\label{eq:nuclear}
\PP_{\hat{\MD}} \partial_{L}\|L\|_{*} = \UM\UM^{\top} \mbox{ for } L\in \hat{\MD}.
\end{equation}
Combining \eqref{eq:l1} and \eqref{eq:nuclear} with the $\nabla h_N$ term, we arrive at an equivalent form of the first order condition, that is, there exists $(S,L)\in\mathcal{B}$ satisfying
\begin{eqnarray}
	\PP_{\MS}\nabla h_N(S+L)+\gamma_N \mathrm{sign}(\OO(S^*))=\mathbf{0}_{J\times J},\notag\\
	\PP_{\hat{\MD}} \nabla h_N(S+L)+\delta_N \UM\UM^{\top}=\mathbf{0}_{J\times J}.\notag
\end{eqnarray}
We will show the existence and uniqueness of the solution to the above equations using contraction mapping theorem. We first construct the contraction operator.
Let $(S,L)=(S^*+\Delta_S,L^*+\Delta_L)$.
We plug \eqref{eq:taylor-nabla} into the above equations, and arrive at their equivalent ones
\begin{eqnarray}
\PP_{\MS}\II(\Delta_S+\Delta_L)+\PP_{\MS}\nabla{R}_N(\Delta_S+\Delta_L)+\gamma_N \mathrm{sign}(\OO(S^*))=\mathbf{0}_{J\times J},\notag\\
\PP_{\hat{\MD}}\II(\Delta_S+\Delta_L)+\PP_{\hat{\MD}}\nabla{R}_N(\Delta_S+\Delta_L)+\delta_N \UM \UM^{\top} = \mathbf{0}_{J\times J}.\label{eq:equivalent}
\end{eqnarray}
We define an operator $\tilde{\FF}_{\hat{\MD}}: \MS\times (\hat{\MD}-L^*) \to \MS\times \hat{\MD}$,
$$
\tilde{\FF}_{\hat{\MD}}(\Delta_S,\Delta_L)= \Big(\PP_{\MS}\II(\Delta_S+\Delta_L),\PP_{\hat{\MD}}\II(\Delta_S+\Delta_L)\Big),
$$
where the set  $\hat{\MD}-L^*=\{L-L^*:L\in \hat{\MD} \}$.
We  further transform equation \eqref{eq:equivalent} to
\begin{equation}\label{eq:equivalent-equation2}
\begin{aligned}
	&\tilde{\FF}_{\hat{\MD}} (\Delta_S,\Delta_L) + (\PP_{\MS}\nabla{R}_N(\Delta_S+\Delta_L),\PP_{\hat{\MD}}\nabla{R}_N(\Delta_S+\Delta_L))+ (\gamma_N \mathrm{sign}(\OO(S^*)),\delta_N \UM \UM^{\top})\\
=&(\mathbf{0}_{J\times J},\mathbf{0}_{J\times J}).
\end{aligned}
\end{equation}
Notice that the projection $\PP_{\hat{\MD}}$ is uniquely determined by the matrix $\UM$. The next lemma states that the mapping
$\UM\to\PP_{\hat{\MD}} $ is Lipschitz.
\begin{lemma}\label{lemma:lip}
	We write the eigendecomposition $L= U_1 D_1 U_1^{\top}$, and define the corresponding linear spaces $T_{L}\mathfrak{L}$ and $\mathcal{D}$ as \begin{equation}\label{eq:def:D-tangent}
\mathcal{D}=\{U_1 D'_1 U_1^{\top}: D'_1 \mbox{ is a $K\times K$ diagonal matrix}\},
	\end{equation} and
	\begin{equation}\label{eq:def:L-tangent}
T_{L}\mathfrak{L}=\{U_1 Y+Y^{\top} U_1^{\top}: Y \mbox{ is a $K\times J$ matrix} \}.
	\end{equation}
	 Then, the mappings $U_1\to \PP_{{T_{L}\mathfrak{L}}}$, $U_1\to \PP_{{T_{L}\mathfrak{L}}^{\bot}}$ and $U_1\to \PP_{{\mathcal{D}}}$ are Lipschitz in $U_1$. That is, for all $J\times J$ symmetric matrix $M$, there exists a constant $\kappa$ such that
	\begin{equation*}
\begin{aligned}
	&\max\Big\{\|\PP_{{T_{L}\MFL}}M-\PP_{{T_{L^*}\MFL}}M\|_{\infty},
	\|\PP_{{T_{L}\MFL}^{\bot}}M-\PP_{{T_{L^*}\MFL}^{\bot}}M\|_{\infty},
	\|\PP_{{\mathcal{D}}}M-\PP_{{\mathcal{D^*}}}M\|_{\infty}\Big\}\\
\leq & \kappa \|U_1-U_1^*\|_{\infty}\|M\|_{\infty}.
\end{aligned}
	\end{equation*}
\end{lemma}
Similar to Lemma~\ref{lemma:f-inv}, under Assumption A1 and A3, we have that $\tilde{\FF}_{\MD^*}$ is invertible, where we define $\tilde{\FF}_{\MD^*}:\MS\times \mathcal{D^*} \to  \MS\times \mathcal{D^*},$
\begin{equation}\label{eq:def:tilde-F}
\tilde{\FF}_{\mathcal{D}^*}(S',L') = (\PP_{\MS}\{\II(S'+L')\}, \PP_{\mathcal{D^*}}\{\II(S'+L')\}), \mbox{ for $S'\in \MS$ and $L'\in \mathcal{D^*}$, }
\end{equation}
and 
$\mathcal{D}^*=\{U^*_1 D'_1 U_1^{*\top}: D'_1 \mbox{ is a $K\times K$ diagonal matrix}\}$.
According to the invertibility of $\tilde{\FF}_{\mathcal{D}^*}$, Lemma~\ref{lemma:lip} and  the fact that $\|\UM-U_1^*\|_{\infty}\leq \cu$,
we know that $\tilde{\FF}_{\hat{\MD}}$ is also invertible over $\MS\times (\hat{\MD}-L^*)$ and  is Lipschitz in $\UM$ for sufficiently large $N$.
We apply $	\tilde{\FF}_{\hat{\MD}}^{-1}$ on both sides of \eqref{eq:equivalent-equation2} and transform it to a fixed point problem,
\begin{equation}\label{eq:fix-d}
(\Delta_{S},\Delta_L) = \CC(\Delta_S,\Delta_L),
\end{equation}
where the operator $\CC$ is defined by
\begin{equation}\label{eq:contraction-d}
\begin{aligned}
&\CC(\Delta_S,\Delta_L) \\
= &- 	\tilde{\FF}_{\hat{\MD}}^{-1}\Big((\PP_{\MS}\nabla{R}_N(\Delta_S+\Delta_L),\PP_{\hat{\MD}}\nabla{R}_N(\Delta_S+\Delta_L))+(\gamma_N \mathrm{sign}(\OO(S^*)),\delta_N \UM \UM^{\top})\Big).
\end{aligned}
\end{equation}
Define the set $\mathcal{B}^*=\mathcal{B}-(S^*,L^*)=\{(S-S^*,L-L^*):(S,L)\in \mathcal{B} \}$.
We will show that with a probability converging to $1$, $\CC$ is a contraction mapping over $\mathcal{B^*}$.
First, according to \eqref{eq:nabla-remain-bound} and the definition of set $\mathcal{B}$, it is easy to check that with probability converging to $1$, $\CC(\Delta_S,\Delta_L)\in \mathcal{B}^*$ for all $(S,L)\in \mathcal{B}^*$, so $\CC(\mathcal{B}^*)\subset \mathcal{B}^*$.
Next, according to \eqref{eq:nabla-2-remain-bound}  and the boundedness of $\tilde{\FF}_{\hat{\MD}}^{-1}$, we know that $\CC(\Delta_S,\Delta_L)$ is Lipschitz in $(\Delta_S,\Delta_L)$ with a probability converging to $1$. To see the size of the Lipschitz constant, according to \eqref{eq:nabla-2-remain-bound} we know that $\nabla R_N(\Delta_S+\Delta_L)$ is Lipschitz with respect to $(\Delta_S,\Delta_L)$ with the Lipschitz constant of order $O_P(\cu)$. Therefore, the Lipschitz constant for $\CC$ is also of order $O_P(\cu)$.
Consequently, $\CC$ is a contraction mapping over the complete metric space $\mathcal{B}^*$ with a probability converging to $1$. According to the Banach fixed point theorem \citep{debnath2005introduction}, \eqref{eq:fix-d} has a unique solution in $\mathcal{B}^*$ with a probability converging to $1$. This concludes our proof.
\end{proof}

%

\begin{proof}[Proof of Lemma~\ref{lemma:trans}]
According to Assumption A3, $\MS\cap T_{L^*}\MFL=\mathbf{0}_{J\times J}$. Then, for all $L\in T_{L^*}\MFL$ and $L\neq \mathbf{0}_{J\times J}$, we have
$$
\|L-\PP_{\MS}L \|_F>0,
$$
where the norm $\|\cdot\|_F$ is the Frobenius norm.
Because the set $\{L:\|L\|_F=1\}$ is compact, $\inf_{\|L\|_F=1,L\in T_{L^*}\MFL} \|L-\PP_{\MS}L \|_F>0$.
Taking $\varepsilon = \inf_{\|L\|_F=1,L\in T_{L^*}\MFL} \|L-\PP_{\MS}L \|_F$,
\begin{equation}\label{eq:boundp}
\|L-\PP_{\MS}L \|_F\geq \varepsilon\|L\|_F, \mbox{ for all $L\in T_{L^*}\MFL$. }
\end{equation}
For a lower bound for $\|L+S\|_F$,
we have
$$
\|S+L\|_F= \|\PP_{\MS^{\bot}}L \|_F +\|\PP_{\MS}(L+S) \|_F\geq \|\PP_{\MS^{\bot}}L \|_F = \|L-\PP_{\MS}L \|_F\geq \varepsilon\|L\|_F,
$$
for $L\in T_{L^*}\MFL$ and $S\in\MS$,
where the last inequality   is due to \eqref{eq:boundp}. We complete the proof by noticing all norms are equivalent for finite dimensional spaces.
\end{proof}

\begin{proof}[Proof of Lemma~\ref{lemma:L}]
Taking $Y=D_1^*(\UM-U_1^*)^{\top}+\frac{1}{2}(\DDD-D_1^*) U_1^{*\top}$, we have
$
\Delta_L=
U^*_1 Y+Y^{\top} U_1^{*\top}.
$ Therefore, $\Delta_L\in T_{L^*}\MFL$ and (i) is proved.
Let
$
\Delta_{U_1}=\UM-U_1^*$  and  $\Delta_{\DM}= \DDD-D_1^*$.
We have
$$
\LD-L^*= (U_1^*+\Delta_{U_1})(D_1^*+\Delta_{D_1})(U_1^*+\Delta_{U_1})^{\top}- U_1^* D_1^* U_1^{*\top}= \Delta_L + O(\|\Delta_{U_1}\|_{\infty}^2+\|\Delta_{D_1} \|_{\infty}^2).
$$
According to Lemma~\ref{lemma:d} and \eqref{eq:u-bound} we have $O(\|\Delta_{U_1}\|_{\infty}^2+\|\Delta_{D_1} \|_{\infty}^2)\leq \kappa \cusq $. Thus, (iii) is proved.
To prove (ii), we need the following eigenvalue perturbation result.
\begin{lemma}[Eigenvalue perturbation]	\label{lemma:eigen-perterb}
	Under Assumption A2, for all $J\times K$ matrix  $~U_1$ such that $U_1^{\top} U_1 = I_K $, and $\|U_1-U^*\|_{\infty}=\cu$, and all $K\times K$ diagonal matrix $D_1$ such that $\|D_1-D_1^*\|_{\infty}\leq \kappa \cu $, there exists a positive constant $\varepsilon$  independent with $U_1$ and $D_1$ (possibly depending on $\kappa$) satisfying
	$$
	\|U_1 D_1 U_1^{\top}-L^*\|_{\infty}\geq \varepsilon \cu.
	$$
\end{lemma}
As a direct application of the above lemma, we have
\begin{equation}\label{eq:ld-low}
 \|\LD-L^*\|_{\infty}\geq \varepsilon \cu.
\end{equation}
 Combing (iii) with \eqref{eq:ld-low}, we have (ii) proved.
\end{proof}

\begin{proof}[Proof of Lemma~\ref{lemma:solution-approx}]
	Assume that on the contrary, \eqref{eq:optim-m2} has two solutions $(\SM,\LM)$ and $(\SA,\LA)$. Similar to $(\SM,\LM)$, $(\SA,\LA)$ also satisfy  \eqref{eq:s-bound}, \eqref{eq:d-bound} and \eqref{eq:u-bound} if we replace $(\SA,\LA)$ by  $(\SM,\LM)$, and it is also an interior point of $\MM_2$. Thus, it satisfies the first order condition of \eqref{eq:optim-m2}.  That is,
	\begin{eqnarray}\label{eq:first-order-FF}
	\PP_{\MS}\nabla h_N(\SM+\LM)+\gamma_N \mathrm{sign}(\OO(S^*))&=&\mathbf{0}_{J\times J}\notag\\
	\PP_{{T_{\LM}\MFL}} \nabla h_N(\SM+\LM)+\delta_N \UM\UM^{\top}&=&\mathbf{0}_{J\times J}.\notag
	\end{eqnarray}
	We define an operator $\FF_{L}: \MS\times T_{L}\mathfrak L \to  \MS\times T_{L}\mathfrak L$ in a similar way as that of $\FF$,
	 \begin{equation}
	 \FF_L(S,L') = (\PP_{\MS}\{ \II(S+L')\}, \PP_{T_{L}\mathfrak L}\{ \II(S+L')\}).
	 \end{equation}
	With similar arguments as those leading towards \eqref{eq:fix-d}, we know that $\FF_{\hat{L}_{\MM_2}}$ is invertible with the aid of Lemmas~\ref{lemma:f-inv} and \ref{lemma:lip}, and \eqref{eq:first-order-FF} is transformed to
	\begin{equation}\label{eq:first-order}
\begin{aligned}
	(\Delta_{\SM},\Delta_{\LM})= \FF_{\hat{L}_{\MM_2}}^{-1}\Big(&(\PP_{\MS}\nabla{R}_N(\Delta_{\SM}+\Delta_{\LM}),\PP_{{T_{\hat{L}_{\MM_2}}\MFL}}  \nabla{R}_N(\Delta_{\SM}+\Delta_{\LM}))\\
& +(\gamma_N \mathrm{sign}(\OO(S^*)),\delta_N \UM \UM^{\top})\Big),
\end{aligned}
	\end{equation}
	where
	$\Delta_{\SM}=\SM-S^*$ and $\Delta_{\LM}=\LM-L^*$.
	Similarly for $(\SA,\LA)$, we have
		\begin{equation}\label{eq:first-order-tilde}
        \begin{aligned}
	(\Delta_{\SA},\Delta_{\LA})= {\FF}_{\tilde{L}_{\MM_2}}^{-1}\Big(&(\PP_{\MS}\nabla{R}_N(\Delta_{\SA}+\Delta_{\LA}),\PP_{T_{\tilde{L}_{\MM_2}}\MFL} \nabla{R}_N(\Delta_{\SA}+\Delta_{\LA}))\\
&+(\gamma_N \mathrm{sign}(\OO(S^*)),\delta_N \UA \UA^{\top})\Big).
        \end{aligned}
		\end{equation}
	Similar to the definition \eqref{eq:contraction-d}, for $(S,L)\in \MM_2$ we define
	\begin{equation*}
\begin{aligned}
		&\CC_{(S,L)}(\Delta_S,\Delta_L)\\
=&- 	{\FF}_{L}^{-1}\Big((\PP_{\MS}\nabla{R}_N(\Delta_S+\Delta_L),\PP_{T_{L}\mathfrak{L}}\nabla{R}_N(\Delta_S+\Delta_L))+(\gamma_N \mathrm{sign}(\OO(S^*)),\delta_N U_1 U_1^{\top}\Big),
\end{aligned}
	\end{equation*}
	where $L$ has the eigendecomposition
	$
	L=U_1 D_1 U_1^{\top}
	$
	, $\|U_1-U_1^*\|_{\infty}\leq \cu$ and $\|D_1-D_1^*\|_{\infty}\leq \cd$.
	The operator $\CC_{(S,L)}$ is well defined, because  for $L\in\MM_2$ the eigendeposition of $L$ is uniquely determined given $(U_1,D_1)$ is in the set
	$ \{(U_1,D_1):\|U_1-U_1^*\|_{\infty}\leq \cu \mbox{ and } \|D_1-D_1^*\|_{\infty}\leq \cd\}$. For more results on eigenvalue perturbation, see Chapter 4, \cite{parlett1980symmetric}.
	Now, we take difference between \eqref{eq:first-order} and \eqref{eq:first-order-tilde},
	\begin{equation}\label{eq:difference}
	(\SA-\SM,\LA-\LM)= \CC_{\SA,\LA}(\Delta_{\SA},\Delta_{\LA})-\CC_{\SM,\LM}(\Delta_{\SM},\Delta_{\LM}).
	\end{equation}
	We provide an upper bound for the norm of the right-hand side of the above equation. We split the right-hand side of the above display into two terms to get
	\begin{eqnarray}
	&&\|\CC_{\SA,\LA}(\Delta_{\SA},\Delta_{\LA})-\CC_{\SM,\LM}(\Delta_{\SM},\Delta_{\LM})\|_{\infty}\notag\\
	&\leq &
	\|\CC_{\SA,\LA}(\Delta_{\SA},\Delta_{\LA})-\CC_{\SA,\LA}(\Delta_{\SM},\Delta_{\LM})\|_{\infty}\label{eq:bound1}\\
	&&+	\|\CC_{\SA,\LA}(\Delta_{\SM},\Delta_{\LM})-\CC_{\SM,\LM}(\Delta_{\SM},\Delta_{\LM})\|_{\infty}\label{eq:bound2}
	\end{eqnarray}
	We present upper bounds for \eqref{eq:bound1} and \eqref{eq:bound2} separately.
	For \eqref{eq:bound1}, using similar arguments as those in the Proof of Lemma~\ref{lemma:d}, we have that with probability converging to $1$, $\CC_{\SA,\LA}(\cdot,\cdot)$ is a Lipschitz operator  with an $O(\cu)$ Lipschitz constant, that is,
	$$
	\eqref{eq:bound1}\leq_P \kappa \cu\times \max\Big(
	\|\SM-\tilde{S}_{\MM_2}\|_{\infty}, \| \LM-\tilde{L}_{\MM_2} \|_{\infty}
	\Big).
	$$
	We proceed to an upper bound of \eqref{eq:bound2}. Thanks to the Lipschitz property of $\FF_L$ and $\PP_{T_{L}\MFL}$ and the invertibility of $\FF_{L^*}$, with a probability converging to $1$, $\CC_{(S,L)}(\Delta_\SM,\Delta_\LM)$ is Lipschitz in $(S,L)$ when $(\Delta_\SM,\Delta_\LM)$ is fixed. Moreover, according to \eqref{eq:nabla-remain-bound}, $\|\Delta_\SM\|_{\infty}\leq \cs$ and $\|\Delta_{\LM} \|_{\infty}\leq \cu$, we have
	\begin{equation}\label{eq:eqre}
	\|\nabla{R}_N(\Delta_{\SA}+\Delta_{\LA})\|_{\infty} \leq O_P(\frac{1}{\sqrt{N}}).
	\end{equation}
	Combining \eqref{eq:eqre} with the fact that $\FF_{L}$ and $\PP_{T_{L}\MFL}$ are locally Lipschitz in $L$, we have that
	$$
	\eqref{eq:bound2}\leq_P \kappa\gamma_N\times\max\Big(
	\|\SM-\tilde{S}_{\MM_2}\|_{\infty}, \| \LM-\tilde{L}_{\MM_2} \|_{\infty}\Big).
	$$
	We combine the upper bounds for \eqref{eq:bound1} and \eqref{eq:bound2} with the equation \eqref{eq:difference} to get
	$$
	\max\Big(
	\|\SM-\tilde{S}_{\MM_2}\|_{\infty}, \| \LM-\tilde{L}_{\MM_2} \|_{\infty}
	\Big)\leq_P 2\kappa \cu \times\max\Big(
	\|\SM-\tilde{S}_{\MM_2}\|_{\infty}, \| \LM-\tilde{L}_{\MM_2} \|_{\infty}
	\Big).
	$$
	Consequently,
	$
	\SM=_P\tilde{S}_{\MM_2}\mbox{ and }\LM=_P\tilde{L}_{\MM_2}.
	$
	We proceed to prove \eqref{eq:solution-approx}. According to \eqref{eq:first-order} and \eqref{eq:eqre}, we have
	$$
	(\SM-S^*,\LM-L^*)=\gamma_N \FF_{\hat{L}_{\MM_2}}^{-1}(\qq_{\UM})+o_P(\gamma_N),
	$$
	where $\qq_{\UM}=(\mathrm{sign}(\OO(S^*)),\rho \UM \UM^{\top})$.
	 Because both $\FF_{\hat{{L}}_{\MM_2}}$  and $\qq_{\UM}$ are Lipschitz continuous in $\UM$, and $\|\UM-U_1^*\|\leq \cu$, we have
	$$
	(\SM-S^*,\LM-L^*)=\gamma_N\FF^{-1}\qq_{U_1^*}+o_P(\gamma_N).
	$$
\end{proof}

\begin{proof}[Proof of Lemma~\ref{lemma:equivalent}]
According to \eqref{eq:subdifferntial-s}, we have
\begin{equation*}
\begin{aligned}
&\PP_{\MS^{\bot}}\partial_S H|_{(\SM,\LM)} \\
=&\{ \PP_{\MS^{\bot}}\nabla{h}_N(\SM+\LM)+\gamma_N \PP_{\MS^{\bot}}(\mathrm{sign}(\OO(S^*)))+\gamma_N \PP_{\MS^{\bot}}W: \|W\|_{\infty}\leq 1 \mbox{ and } W\in \MS^{\bot}\}.
\end{aligned}
\end{equation*}
Notice that $\PP_{\MS^{\bot}}(\mathrm{sign}(\OO(S^*)))=\mathbf{0}_{J\times J}$ and $\PP_{\MS^{\bot}}W=W$ for $W\in\MS^{\bot}$. Therefore, we have
\begin{eqnarray*}
& & \mathbf{0}_{J\times J}\in  \PP_{\MS^{\bot}}\partial_S H|_{(\SM,\LM)}\\
&\iff & \exists W \in \MS \mbox{ such that } \|W\|_{\infty}\leq 1\mbox{ and } \PP_{\MS^{\bot}}\nabla{h}_N(\SM+\LM) = -\gamma_N W. \\
&\iff & \|\PP_{\MS^{\bot}}\nabla{h}_N(\SM+\LM)\|_{\infty}\leq \gamma_N.
\end{eqnarray*}
Similarly, according to \eqref{eq:subdifferential-l}, we have
\begin{eqnarray*}
\PP_{(T_{\LM}\MFL)^{\bot}}\partial_L H|_{(\SM,\LM)} =\{  \PP_{(T_{\LM}\MFL)^{\bot}}\nabla{h}_N(\SM+\LM)
+\delta_N \UCM W \UCM^{\top}: \|W\|_{2}\leq 1 \}.
\end{eqnarray*}
Consequently,
\begin{eqnarray*}
&&\mathbf{0}_{J\times J}\in \PP_{(T_{\LM}\MFL)^{\bot}}\partial_L H|_{(\SM,\LM)}\\
&\iff & \exists W \mbox{ such that } \|W\|_{2}\leq 1\mbox{ and } {\PP_{(T_{\LM}\MFL)^{\bot}}}\nabla{h}_N(\SM+\LM) = -\UCM W \UCM^{\top}. \\
&\iff & \|\PP_{(T_{\LM}\MFL)^{\bot}}\nabla{h}_N(\SM+\LM)\|_{2}\leq \delta_N.
\end{eqnarray*}
These two equivalent expressions concludes our proof.
\end{proof}

\begin{proof}[Proof of Lemma~\ref{lemma:another}]
Let
$
\tilde{S} = \SA + \DSB ,
$
where $\SA\in \MS$ and $\DSB \in \MS^{\bot}$. Let
$
\tilde{L}=\LA+\DLB,
$
where
$
\LA= \UA \DAone \UA^{\bot}  \mbox{ and } \DLB =  \UAtwo \DAtwo \UAtwo^{\bot}.
$
Notice that $(\SA, \LA)\in \mathcal{M}_2$ and $\DLB\in (T_{\tilde{L}}\MFL)^{\bot}$.
Similar to \eqref{eq:subdifferntial-s} and \eqref{eq:subdifferential-l} we have the sub-differentials of $H(S,L)$ at $(\SA,\LA)$,
\begin{equation*}
\partial_S H\vert_{(\SA,\LA)} =\{ \nabla{h}_N(\SA+\LA)+\gamma_N \mathrm{sign}(\OO(S^*))+\gamma_N W_1: \|W_1\|_{\infty}\leq 1, \mbox{ and } W_1\in \MS^{\bot}\},
\end{equation*}
and
\begin{equation*}
\partial_{L}H|_{(\SA,\LA)}=\{\nabla{h}_N(\SM+\LM)+\delta_N\UA\UA^{\top}+\delta_N\UAtwo W_2 \UAtwo^{\top}: \|W_2\|_2\leq 1\}.
\end{equation*}
Let $W_1 = \mathrm{sign}(\DSB)$ and $W_2= \mathrm{sign}(\DAtwo)$ in the above expressions for sub-differentials. According to the definition of sub-differential, we have
\begin{equation*}
\begin{aligned}
&H(\tilde{S},\tilde{L})-H(\SA,\LA)\\
\geq
&\nabla h_N(\SA+\LA)\cdot(\DSB+\DLB) + \gamma_N \|\DSB\|_1 +\delta_N \| \DLB \|_*.
\end{aligned}
\end{equation*}
Because $\DSB\in\MS^{\bot}$ and $\DLB\in (T_{\tilde{L}}\MFL)^{\bot}$, we further expand the above inequality,
\begin{equation}\label{eq:H-lower}
\begin{aligned}
&H(\tilde{S},\tilde{L})-H(\SA,\LA)\\
\geq &
\PP_{\MS^{\bot}}\Big[\nabla h_N(\SA+\LA)\Big]\cdot\DSB+
\PP_{(T_{\tilde{L}}\MFL)^{\bot}}\Big[\nabla h_N(\SA+\LA)\Big]\cdot\DLB \\
& + \gamma_N (\|\DSB\|_1 +\rho \| \DLB \|_*) .
\end{aligned}
\end{equation}
We provide a lower bound for the right-hand side of \eqref{eq:H-lower}.
According to \eqref{eq:gradient-hat}, the definition of $\FF^{\bot}$, and the Lipschitz property of $\PP_{(T_{\tilde{L}}\MFL)^{\bot}} $ according to Lemma~\ref{lemma:lip}, we have
\begin{equation*}\label{eq:lower-likelihood}
\begin{aligned}
&\Big(\PP_{\MS^{\bot}}\nabla h_N(\SA+\LA),\PP_{(T_{\tilde{L}}\MFL)^{\bot}} \nabla h_N(\SA+\LA)\Big)\\
=&\gamma_N \FF^{\bot} \FF^{-1}\Big(\mathrm{sign}(\OO(S^*)), \rho U_1^* U_1^{*\top}\Big)+o_P(\gamma_N).
\end{aligned}
\end{equation*}
According to Assumption A4 and the above expression, we have
\begin{eqnarray}\label{eq:lower-ph}
|\PP_{\MS^{\bot}}\nabla h_N(\SA+\LA)\cdot\DSB |&<_P& \gamma_N \|\DSB \|_{\infty}\notag,\\
|\PP_{(T_{\tilde{L}}\MFL)^{\bot}}\nabla h_N(\SA+\LA)\cdot\DLB| &<_P& \delta_N \|\DLB \|_2.
\end{eqnarray}
We proceed to the $L_1$ penalty term. It has a lower bound
\begin{equation}\label{eq:lower-L1}
\|\DSB\|_1\geq \|\DSB\|_{\infty}.
\end{equation}
For the nuclear norm term, we have
\begin{equation}\label{eq:lower-trace}
\|\DLB\|_*=\|\DAtwo \|_*\geq \|\DAtwo \|_{\infty}=\|\DLB \|_{2}.
\end{equation}
The first equality in \eqref{eq:lower-trace} is due to the definition of $\DLB$. The inequality and second equality hold because $\DAtwo$ is a diagonal matrix and its nuclear norm is the same as $L_1$ norm, and its spectral norm is the same as its maximum norm.
Combineing \eqref{eq:H-lower}, \eqref{eq:lower-ph}, \eqref{eq:lower-L1} and \eqref{eq:lower-trace}, we have
$$
H(\tilde{S},\tilde{L})>_P H(\SA,\LA),
$$
provided $\|\DSB\|_{\infty}> 0$ or $\|\DLB \|_*>0$. Because $(\tilde{S},\tilde{L})$ is a solution to \eqref{eq:optim-m1}, the above statement implies $\DSB=\DLB=\mathbf{0}_{J\times J}$. Therefore, $\tilde{S}=\SA$, $\tilde{L}=\LA$, and $(\tilde{S},\tilde{L})\in\mathcal{M}_2$,
\end{proof}

\begin{proof}[Proof of Lemma~\ref{lemma:lip}]
	We first investigate the tangent space $T_{L}\MFL$,
	\begin{multline}
	T_{L}\MFL=\Big\{
	[U_1,U_2]\begin{bmatrix}
	Y_{11} & Y_{12}\\
	Y_{21} & \mathbf{0}_{(J-K)\times (J-K)}
	\end{bmatrix}
	[U_1,U_2]^{\top}: Y_{11} \mbox{ is symmetric, and } Y_{12} =Y_{21}^{\top}
	\Big\}.
\end{multline}
For any symmetric $M$, let $N=U^{\top} M U$. Then $N$ is symmetric and
$
M=U N U^{\top}.
$
We write $M$ as
$$
M=
[U_1,U_2]\begin{bmatrix}
N_{11} & N_{12}\\
N_{21} & N_{22}
\end{bmatrix}[U_1,U_2]^{\top}.
$$
Therefore, $$\PP_{T_L\MFL} (M) =
[U_1,U_2]\begin{bmatrix}
N_{11} & N_{12}\\
N_{21} & \mathbf{0}_{(J-K)\times (J-K)}
\end{bmatrix}[U_1,U_2]^{\top}, $$
which is Lipschitz in $[U_1,U_2]$.
Because $U_2$ is orthogonal to $U_1$,  we could choose $U_2$ such that $U_2$ is also Lipschitz in $U_1$. As a result, the operator $\PP_{T_{L}\MFL}$ is Lipschitz in $U_1$.
%
%
For the operator $\PP_{{T_L\MFL}^{\bot}}$, we have
$$\PP_{{T_L\MFL}^{\bot}} =
[U_1,U_2]\begin{bmatrix}
	\mathbf{0}_{K\times K} & \mathbf{0}_{J\times (J-K)}\\
	\mathbf{0}_{(J-K)\times K} & N_{22}
\end{bmatrix}[U_1,U_2]^{\top}, $$
for the same symmetric matrices $M$ and $N$ discussed before. Thus, the above display is also Lipschitz in $U_1$.
Similarly for $\PP_{\mathcal{D}}$, we have
$$
\mathcal{D}= \{U_1D_1U_1^{\top}: D_1\mbox{ is a $K\times K$ diagonal matrix} \}.
$$
For the same symmetric matrices $M$ and $N$, we have
$$\PP_{\mathcal{D}} M =
[U_1,U_2]\begin{bmatrix}
\textrm{diag}(N_{11}) & \mathbf{0}_{(J-K)\times K}\\
\mathbf{0}_{J\times (J-K)} & \mathbf{0}_{(J-K)\times (J-K)}
\end{bmatrix}[U_1,U_2]^{\top}, $$
where $\textrm{diag}(N_{11})$ is the diagonal components of $N_{11}$.
It is also Lipschitz continuous in $U_1$.
\end{proof}

\begin{proof}[Proof of Lemma~\ref{lemma:eigen-perterb}]
Applying Fact 10, p.~15-3, \cite{hogben2006matrix}, with $A=L^*$ and $\tilde{A}=U_1D_1U_1^{\bot}$,  we have that, under Assumption A2, $
 \|U_1-U_1^*\|_{\infty}\leq \kappa \|U_1D_1U_1^{\bot}-L^*\|_2,
$
for some $\kappa$ and sufficiently small $\|D_1-D_1^*\|_{\infty}$.
Because all norms for finite dimensional space are equivalent, this inequality leads to
$$ \|U_1-U_1^*\|_{\infty}\leq \kappa \|U_1D_1U_1^{\bot}-L^*\|_{\infty},
$$
for a possibly different $\kappa$.
Recall the assumptions of the lemma that $\|U_1-U_1^*\|_{\infty}=\cu$. Therefore,
$$
\|U_1D_1U_1^{\bot}-L^*\|_{\infty}\geq \kappa^{-1}\cu.
$$
\end{proof}
\end{document}